%% file: main.tex
\documentclass[acmtog]{acmart}

\AtBeginDocument{%
  }

\acmSubmissionID{2393}

\usepackage{listings}
\usepackage{xcolor}
\usepackage{multirow}
\usepackage{colortbl}
\usepackage{tabularx}
\usepackage{xspace}
\usepackage{enumitem}
\usepackage{subcaption}
\usepackage{tikz}
\usepackage{booktabs}
\usepackage{multirow}
\usepackage{siunitx}

\usetikzlibrary{arrows.meta, positioning, fit, calc}

\setcopyright{none}

\AtBeginDocument{%
  }

\makeatletter
\@ACM@balancefalse
\@ACM@pbalancefalse
\makeatother

\newcommand{\method}{MaterialClusterGS\xspace}

\setcopyright{none}
\copyrightyear{2026}
\acmYear{2026}
\acmDOI{10.1145/3592405}
\acmJournal{TOG}
\acmVolume{42}
\acmNumber{4}
\acmArticle{1}
\acmMonth{8}

\title{MaterialClusterGS: Palette-Based Material Decomposition and Physically-Based Relighting with 2D Gaussian Splatting}
\author{Hao Zhang}
\orcid{0009-0000-1145-783X}
\affiliation{%
  \institution{Peking University}
  \city{Beijing}
  \country{China}}
\email{zhanghao25@stu.pku.edu.cn}

\author{Ang Li}
\affiliation{%
  \institution{Peking University}
  \city{Beijing}
  \country{China}}
\email{psnalaalansp@gmail.com}

\author{Boyan Du}
\affiliation{%
  \institution{University of Science and Technology of China}
  \country{China}}
\email{boyer@mail.ustc.edu.cn}

\author{Junke Zhu}
\affiliation{%
  \institution{University of Science and Technology of China}
  \country{China}}
\email{junkezhu@mail.ustc.edu.cn}

\author{Fei Zhu}
\affiliation{%
  \institution{Peking University}
  \city{Beijing}
  \country{China}}
\email{feizhu@pku.edu.cn}

\author{Meng Gai}
\affiliation{%
  \institution{Peking University}
  \city{Beijing}
  \country{China}}
\email{gaimeng@pku.edu.cn}

\author{Zhangjin Huang}
\affiliation{%
  \institution{University of Science and Technology of China}
  \country{China}}
\email{zhuang@ustc.edu.cn}

\author{Guoping Wang}
\affiliation{%
  \institution{Peking University}
  \city{Beijing}
  \country{China}}
\email{wgp@pku.edu.cn}

\author{Sheng Li}
\affiliation{%
  \institution{Peking University}
  \city{Beijing}
  \country{China}}
\email{lisheng@pku.edu.cn}

\renewcommand{\shortauthors}{Zhang et al.}

\begin{document}
\sloppy
\makeatletter
\setlength{\@fptop}{0pt}
\setlength{\@fpbot}{0pt plus 1fil}
\makeatother

\input{chapters/abstract}

\begin{teaserfigure}
  \centering
  \includegraphics[width=\textwidth]{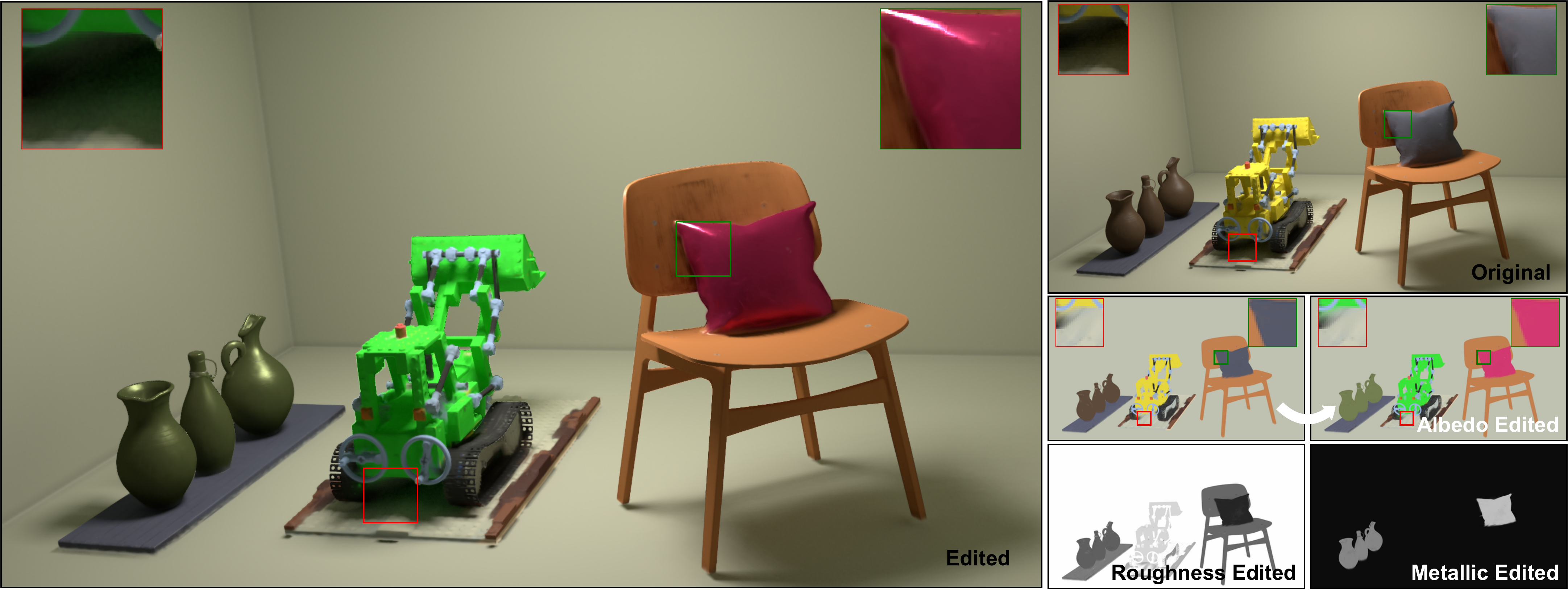}
    \caption{Scene editing with \method. Left: the edited scene rendered with modified palette materials. Right top: the original render before editing. Right bottom: albedo maps before (left) and after (right) material editing, along with the edited roughness and metallic maps. Shared BRDF prototypes ensure consistent materials, enabling edits to propagate coherently. The green inset highlights specular reflections from the low-roughness, high-metallicity edit. The red inset shows path-traced indirect illumination induced by the material edit, rather than albedo-baked appearance changes.}
  \label{fig:teaser}
\end{teaserfigure}




\maketitle

\input{chapters/intro}
\input{chapters/relatedworks}
\input{chapters/motivation}
\input{chapters/methods}
\input{chapters/results}
\input{chapters/discussion}

\input{chapters/conclusion}

\clearpage
\bibliographystyle{ACM-Reference-Format}
\bibliography{chapters/reference}

\input{chapters/images}

\input{chapters/supplement}

\end{document}

%% file: chapters/abstract.tex
\begin{abstract}
We present \method, a palette-based material decomposition framework for 2D Gaussian Splatting that enables physically based relighting and material editing. Existing Gaussian inverse rendering methods typically assign
independent BRDF parameters to individual primitives. While flexible, this local fitting strategy makes material recovery highly under-constrained: shadows, indirect illumination, geometric errors, and visibility residuals can be absorbed into thousands of slightly different local material estimates. Meanwhile, recent palette-based appearance methods operate solely in RGB space without modeling physical materials or illumination. To bridge this gap, we represent scene materials using a compact global palette of shared BRDF prototypes assigned via a continuous spatial material field. Without shared material structure, editing one region does not propagate consistently to others of the same material, making per-primitive decompositions impractical for editing. We jointly optimize the material field, palette prototypes, and environment lighting under a physically based rendering objective. The resulting framework recovers compact, spatially coherent attributes directly usable for material editing, relighting, and transfer.
\end{abstract}

%% file: chapters/intro.tex
\section{INTRODUCTION}

3D Gaussian Splatting~\cite{kerbl2023gaussiansplatting} (3DGS) has become a dominant scene representation for novel view synthesis, combining explicit primitives with efficient differentiable rendering and high image quality. However, appearance is stored as radiance rather than physically meaningful material properties, making inverse rendering and editing difficult.

The main challenge for inverse rendering is not merely that Gaussian colors are non-physical, but that per-Gaussian material degrees of freedom allow each primitive to independently adjust its BRDF parameters. This entanglement is not removed but redistributed: shadows, ambient occlusion, interreflection, and geometric inaccuracies may be baked into local albedo or roughness estimates. While sufficient for reproducing the original scene, the learned appearance no longer behaves like a physically meaningful BRDF---editing materials or relighting produces implausible results because the estimated parameters act as a view-conditioned radiance cache rather than intrinsic material properties.

This issue has motivated two lines of recent work. Inverse rendering methods for Gaussian scenes~\cite{gsir2023,gaussianshader2024,R3DG,chen2025gigs} estimate materials per primitive, yielding spatially inconsistent decompositions that cannot fully decouple residual lighting from albedo. In parallel, palette-based appearance methods~\cite{ren2024palettegaussian,bi2025recolorgaussian,chao2026colorgradedgs,kuang2023palettenerf} cluster scene appearance into a compact set of shared bases, but work purely at the RGB level without decomposing reflectance, geometry, or lighting. Without a shared material structure, neither approach supports reliable downstream editing---modifying one material region does not consistently propagate to all other regions sharing the same material.

In this paper, we propose \method, a palette-based material decomposition framework for 2D Gaussian Splatting~\cite{Huang2DGS2024} that unifies these two directions. We build on 2D Gaussian surfels rather than 3D Gaussians because our physically based rendering requires well-defined ray--surface intersections for visibility and indirect illumination; as noted by IRGS~\cite{IRGS}, 2D Gaussian disks provide unambiguous ray--splat intersections whereas 3D Gaussians do not. Instead of letting every Gaussian invent its own BRDF, we require many Gaussians to share a small set of BRDF prototypes in a global material palette, with each Gaussian receiving a soft assignment predicted from its 3D position. This enforces global material consistency by construction while permitting smooth transitions at boundaries.

To make the recovered materials physically useful, we couple the palette representation with a physically based rendering objective that jointly optimizes materials and illumination. A residual weight estimation stage identifies where the PBR decomposition falls short; a subsequent refinement stage injects per-Gaussian detail guided by the learned weights. The result is an unsupervised decomposition that is compact, editable, and visually faithful.

The main contributions of our work are as follows:
\begin{itemize}[leftmargin=1.5em]
\item We identify excessive material freedom as the key obstacle to editable inverse rendering, and address it with a palette-based material representation that replaces unconstrained per-Gaussian BRDFs with a shared set of BRDF prototypes.
\item We introduce a compact BRDF palette and a continuous spatial material
field that maps 3D positions to shared material assignments, turning material
recovery from local appearance fitting into global material explanation without
pairwise regularization.
% \item We design a multi-stage optimization that couples palette-based materials with a PBR objective, residual weight estimation, and per-Gaussian refinement, enabling material editing, relighting, and transfer.
\item We present a palette-level material editing framework built on the learned
BRDF palette, residual weight estimation, and sparse per-Gaussian refinement,
enabling coherent material editing, material transfer, and forward
path-traced relighting after edits.
\end{itemize}

%% file: chapters/relatedworks.tex
\section{RELATED WORK}

\subsection{Gaussian Splatting and Inverse Rendering}

\paragraph{Gaussian splatting.}
3D Gaussian Splatting (3DGS) represents scenes as collections of anisotropic Gaussians with learnable geometry, opacity, and view-dependent color, enabling high-quality novel view synthesis with real-time performance~\cite{kerbl2023gaussiansplatting}. A key limitation is that appearance is encoded as radiance rather than physically meaningful material parameters, entangling illumination, reflectance, and geometry. 2D Gaussian Splatting~\cite{Huang2DGS2024} replaces volumetric 3D Gaussians with oriented surfels, yielding more accurate surface reconstruction and better compatibility with material-aware rendering. Recent work has further advanced Gaussian-based ray tracing: 3D Gaussian Ray Tracing~\cite{3dgrt2024} introduces efficient tracing of particle scenes for shadows and reflections, and Stochastic Ray Tracing of 3D Transparent Gaussians~\cite{sun2025stochasticraytracingtransparent} accelerates 3DGRT via stochastic ray--Gaussian intersection.

\paragraph{Neural inverse rendering.}
Early neural approaches demonstrated that volumetric representations can be decomposed into shape, reflectance, and lighting: NeRFactor~\cite{zhang2021nerfactor} factors a NeRF into material and illumination, TensoIR~\cite{Jin2023TensoIR} introduces a tensorial decomposition for efficient inverse rendering, and subsequent work models indirect illumination for more accurate decomposition~\cite{zhang2022invrender}. However, MLP-based rendering remains computationally expensive.

\paragraph{Gaussian inverse rendering.}
Gaussian-based methods bring inverse rendering into the splatting setting with real-time potential. Early work such as GS-IR~\cite{gsir2023} and GaussianShader~\cite{gaussianshader2024} estimate per-primitive materials via deferred rendering and shader-style appearance, respectively. A growing body of recent methods tackles relighting and illumination decomposition from complementary angles: GI-GS~\cite{chen2025gigs}, R3DG~\cite{R3DG}, SVG-IR~\cite{sun25svgir}, RNG~\cite{fan25rng}, and IRGS~\cite{IRGS} combine BRDF decomposition with various tracing or neural strategies; DeferredGS~\cite{wu2024deferredgs}, PRTGS~\cite{guo2024prtgs}, RadiosityGS~\cite{jiang2025radiositygs}, and RadioGS~\cite{han2026radiogs} introduce deferred shading, precomputed radiance transfer, radiosity-based light transport, and radiometric consistency for efficient relighting; and GS-ID~\cite{GS-ID}, GlossyGS~\cite{glossygs}, and EnvGS~\cite{xie2024envgs} leverage diffusion priors and environment Gaussians for illumination decomposition and glossy material recovery. Most of these methods rely on per-primitive material prediction, which can produce spatially noisy decompositions; our palette-based formulation addresses this by sharing prototypes through a continuous spatial field.

\subsection{Gaussian Scene Editing and Palette-Based Appearance}

\paragraph{Gaussian scene editing.}
Editable Gaussian methods have developed along several paradigms. Diffusion-based editors render views from the current scene, edit them with a 2D generative prior, and optimize the 3D representation to match~\cite{gaussctrl,vcedit,editsplat}, but must carefully address multi-view consistency. Direct 3D manipulation methods~\cite{gaussianeditor,sceneditor3dgs} modify primitive parameters directly, offering explicit geometric control but requiring careful handling of artifacts. Semantic methods~\cite{semanticgaussians} ground edits through language-conditioned or instance-level selection. Closer to our setting, Editable Reflection Gaussian~\cite{edtable-3dgs} separates diffuse and specular contributions into distinct G-buffers via path-traced Gaussian radiance fields, enabling physically based editing of reflections in reconstructed scenes. However, it focuses on reflection editing rather than full material decomposition and relighting.

\paragraph{Palette-based appearance.}
A useful prior for real scenes is that the number of distinct materials is far smaller than the number of geometric primitives. Palette-based representations exploit this by using a compact set of shared prototypes and a spatial assignment field, improving global consistency, reducing the dimensionality of the inverse problem, and enabling intuitive editing.
PaletteNeRF~\cite{kuang2023palettenerf} first introduced this idea for neural radiance fields. Recent extensions to Gaussian Splatting include PaletteGaussian~\cite{ren2024palettegaussian} for palette-based color editing, RecolorGaussian~\cite{bi2025recolorgaussian} for scene recoloring, and ColorGradedGS~\cite{chao2026colorgradedgs} for palette-driven color grading. However, these methods operate purely in RGB space without modeling physical materials or illumination, limiting them to appearance edits rather than relighting or material transfer, limitations addressed by our palette-of-BRDF formulation.

%% file: chapters/motivation.tex
\section{Preliminary}

\subsection{2D Gaussian Splatting}

3DGS represents a scene as a set of anisotropic 3D Gaussians with learnable geometry, opacity, and view-dependent color encoded as spherical harmonics. 2D Gaussian Splatting~\cite{Huang2DGS2024} replaces these volumetric ellipsoids with oriented 2D surfels defined by a center $\boldsymbol{\mu}$, two tangent vectors $\mathbf{t}_u,\mathbf{t}_v$, a scaling vector $\mathbf{s}$, and an opacity $\alpha$. The local surface normal is $\mathbf{n}=(\mathbf{t}_u\times\mathbf{t}_v)/\|\mathbf{t}_u\times\mathbf{t}_v\|$. During rendering, each surfel is projected to screen space for efficient tile-based rasterization, and the final pixel color is obtained through front-to-back alpha compositing:
\begin{equation}
C(\mathbf{x})=\sum_{i=1}^{N}T_i\,\alpha_i\,c_i,\quad
T_i=\prod_{j=1}^{i-1}(1-\alpha_j),
\end{equation}
where $T_i$ is the accumulated transmittance, $\alpha_i$ is the opacity weighted by the projected 2D density, and $c_i$ is the SH-evaluated color. Because appearance is stored directly as radiance, material properties and illumination are entangled---making inverse rendering and relighting difficult without an explicit material model.

\subsection{Physically Based Gaussian Rendering}

A simplified Disney principled BRDF~\cite{burley12disneybrdf} with the GGX microfacet distribution~\cite{walter2007ggx} and the Schlick Fresnel approximation~\cite{schlick1994brdf} are employed. The outgoing radiance at a surface point is governed by the rendering equation,
\begin{equation}
L_o(\mathbf{p},\boldsymbol{\omega}_o)=\int_{\Omega^+} f(\mathbf{p},\boldsymbol{\omega}_i,\boldsymbol{\omega}_o)\,L_i(\mathbf{p},\boldsymbol{\omega}_i)\,(\mathbf{n}\cdot\boldsymbol{\omega}_i)\,d\boldsymbol{\omega}_i,
\end{equation}
where $f$ is the BRDF, $L_i$ is incoming radiance, and $\mathbf{n}$ is the surface normal. We use a three-parameter material model consisting of albedo $\mathbf{a}$, roughness $r$, and metallic $m$. The BRDF is decomposed into a diffuse (Lambertian) term and a specular (GGX microfacet) term, where the metallic parameter blends between dielectric and metallic Fresnel response.

\section{MOTIVATION}

\begin{figure}[!t]
  \centering
  \includegraphics[width=\columnwidth]{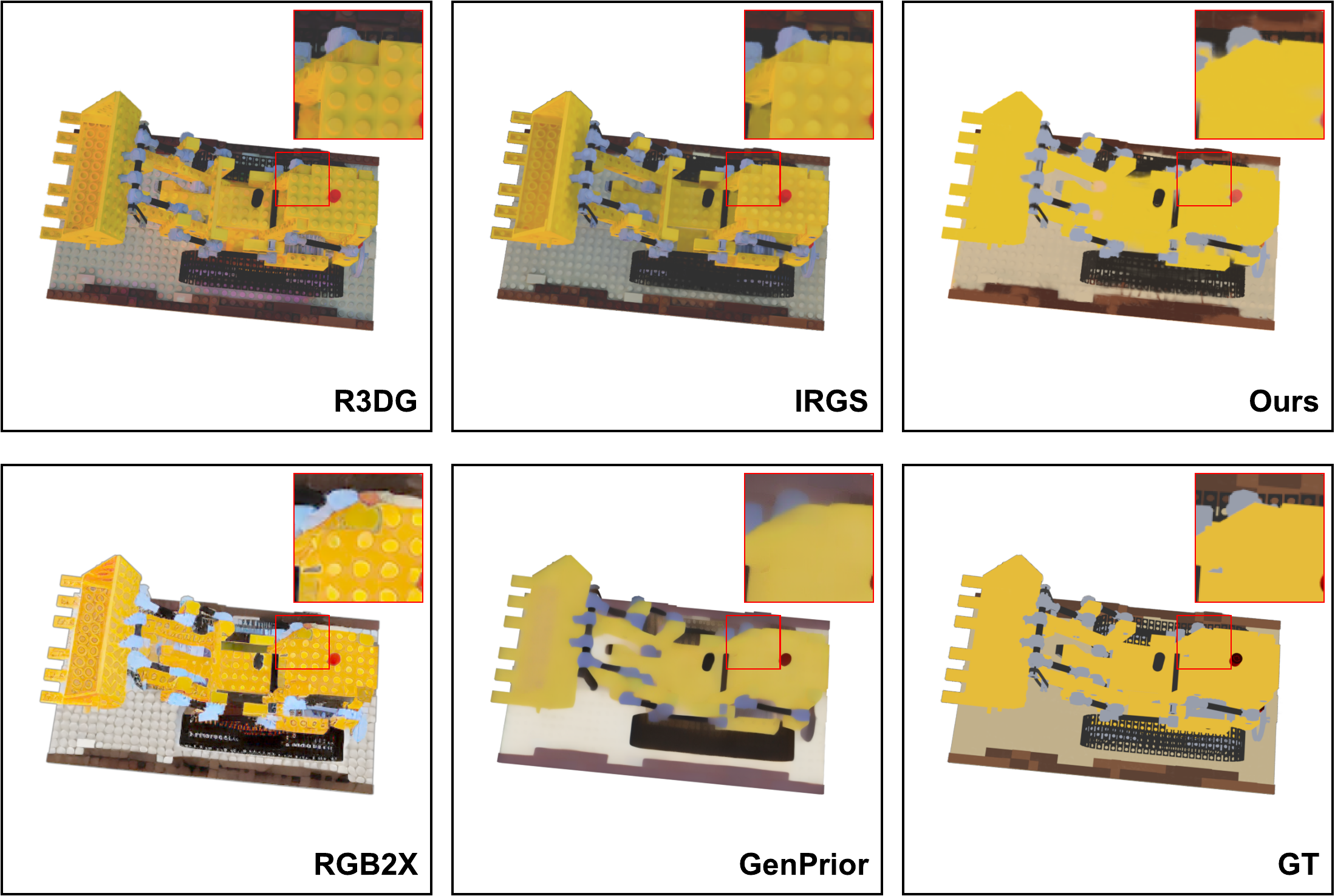}
  \caption{Albedo decomposition comparison. Row~1: inverse rendering methods R3DG and IRGS bake geometric shading artifacts into the albedo, while ours recovers clean results. Row~2: RGB${\leftrightarrow}$X~\cite{rgb2x} exhibits similar baking artifacts, GenPrior from Editable-3DGS~\cite{edtable-3dgs} over-smooths and loses material detail, compared to ground truth.}
  \label{fig:motivation}
\end{figure}

Recovering physically meaningful materials from a pre-trained Gaussian scene is difficult not only because inverse rendering is ill-posed (the same observed radiance can be explained by many combinations of material, geometry, and illumination), but also because the common per-Gaussian material parameterization gives the optimizer too much freedom. Existing inverse rendering methods~\cite{R3DG,IRGS} often let each Gaussian fit its own material parameters. This creates a convenient escape route during optimization: whenever geometry, visibility, or illumination is imperfect, the error can be hidden inside local albedo, roughness, or metallic values. The resulting decomposition may reproduce images well, but it no longer provides a reliable material description for editing or relighting. 
This fragmentation prevents unified editing: modifying what should be a single material requires adjusting thousands of uncoordinated per-primitive parameters. Moreover, imperfect reconstructed geometry introduces further decomposition errors: contact shadows, ambient occlusion, and interreflection that arise from inaccurate surface coverage or misaligned normals are absorbed into the estimated albedo rather than correctly attributed to geometry (see $1^{st}$ row in~\autoref{fig:motivation}, by R3DG and IRGS). To constrain this ambiguity, recent methods incorporate external priors such as diffusion models~\cite{edtable-3dgs,rgb2x} or learned material statistics. However, these priors still introduce artifacts: RGB${\leftrightarrow}$X~\cite{rgb2x} suffers from similar baking artifacts, while Editable-3DGS~\cite{edtable-3dgs} tends to over-smooth decomposition and loses fine material detail (see $2^{nd}$ row in~\autoref{fig:motivation}, by RGB2X and GenPrior).

Instead of relying on external generative priors, we exploit a simple observation inherent to real-world scenes: the number of distinct materials is usually far smaller than the number of geometric primitives. We therefore represent scene appearance using a compact palette of shared BRDF prototypes assigned through a continuous spatial field, encouraging coherent multi-view material explanations while remaining lightweight and scene-adaptive.

%% file: chapters/methods.tex
\section{METHODS}

\begin{figure*}[t]
\centering
\includegraphics[width=0.9\textwidth]{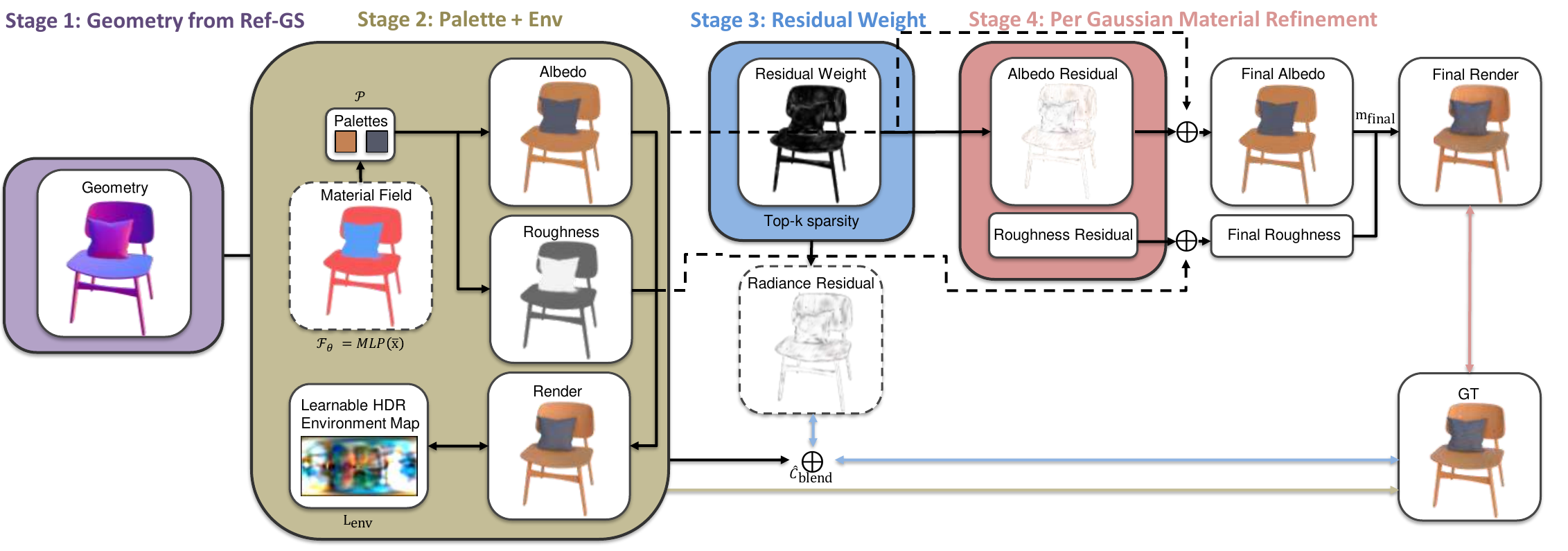}
\caption{Overview of \method. Starting from a pre-trained 2DGS scene (Stage~1), Stage~2 introduces a learnable material palette $\mathcal{P}$ and a spatial material field $\mathcal{F}_\theta$ that assigns each Gaussian to palette entries via softmax with temperature annealing, along with an HDR environment map $L_{\mathrm{env}}$ for physically based rendering. Stage~3 jointly trains per-Gaussian residual weights $w_i$ and the SH radiance field to identify regions where the PBR decomposition is insufficient, using an intermediate blended rendering $\hat{C}_{\mathrm{blend}}$. Stage~4 freezes the learned weights and refines per-Gaussian direct materials via weight-guided blending: $\mathbf{m}_i^{\mathrm{final}}=(1{-}w_i)\mathbf{m}_i^{\mathrm{palette}}+w_i\sigma(\hat{\mathbf{d}}_i)$.}
\label{fig:pipeline}
\end{figure*}

\subsection{Overview}

We assume a pre-trained 2D Gaussian Splatting scene provided by Ref-GS~\cite{zhang2024ref} as the geometric and radiance initialization. Starting from this scene, materials are represented as shared scene-level attributes rather than independent BRDFs attached to individual Gaussians. The scene appearance is explained through a compact palette of shared BRDF prototypes jointly assigned to groups of Gaussians. Our method proceeds through four stages.

The core of the pipeline is a \emph{material palette} $\mathcal{P}=\{\mathbf{m}_k\}_{k=1}^{K}$, a compact set of shared BRDF prototypes. A spatial material field $\mathcal{F}_{\theta}$ maps each Gaussian's 3D position to a soft palette assignment, ensuring that nearby Gaussians receive similar materials without explicit pairwise regularization. Per-Gaussian material residuals $\Delta\mathbf{m}_i$ capture fine local variation beyond what the palette alone can express. A learnable HDR environment map $L_{\mathrm{env}}$ provides the illumination, a residual weight estimation stage identifies Gaussians where the PBR decomposition is insufficient, and a final per-Gaussian refinement stage injects local detail guided by the learned weights.

\subsection{Palette-Based Material Decomposition}

\paragraph{Material parameterization.}
Each palette entry stores albedo, roughness, and metallic parameters in raw form,
\begin{equation}
\mathbf{m}_k=[\hat{\mathbf{a}}_k,\hat{r}_k,\hat{m}_k] \in \mathbb{R}^{5}.
\end{equation}
We activate these with bounded nonlinearities to impose physically meaningful ranges:
\begin{align}
\mathbf{a}_k &= \sigma(\hat{\mathbf{a}}_k) \cdot a_s + a_b, &
r_k &= \sigma(\hat{r}_k), &
m_k &= \sigma(\hat{m}_k),
\end{align}
where $\sigma$ is the sigmoid function, $a_s$ is the albedo scale, and $a_b$ is the albedo bias. The resulting albedo range $[a_b,\, a_s{+}a_b]$ prevents degenerate solutions where albedo approaches zero to absorb missing lighting, while accommodating common bright materials; the metallic parameter is fixed during training to avoid unstable convergence and degenerate metallic--albedo coupling.

\paragraph{Spatial material field.}
For a Gaussian at position $\mathbf{x}_i$, the material is computed as a soft combination of palette entries:
\begin{equation}
\mathbf{m}_i = \sum_{k=1}^{K} w_{ik}\,\mathbf{m}_k, \quad w_{ik} = [\mathbf{w}_i]_k,
\end{equation}
where $\mathbf{w}_i \in \mathbb{R}^K$ is the assignment weight vector produced by a lightweight MLP with sinusoidal positional encoding over normalized 3D coordinates,
\begin{equation}
\mathbf{w}_i = \operatorname{softmax}\left(\frac{\mathcal{F}_{\theta}(\mathbf{x}_i)}{\tau}\right),
\end{equation}
where $\mathcal{F}_{\theta}(\mathbf{x})=\operatorname{MLP}(\gamma(\bar{\mathbf{x}}))$, $\bar{\mathbf{x}}$ is the position normalized within the scene bounding box, and $\gamma(\cdot)$ is a positional encoding with $L$ frequency bands yielding an input dimension of $3+6L$. The temperature $\tau$ is annealed during training: $\tau^{(t+1)} = \max(\tau_{\min},\, \tau^{(t)} - \delta_\tau)$, starting from soft multi-material blending and gradually converging toward confident assignments. Since the MLP is continuous in 3D position, nearby Gaussians naturally receive similar weights, inducing spatial coherence without pairwise regularization.

\paragraph{Material residuals.}
To capture fine local variation beyond the palette, we introduce per-Gaussian material residuals $\Delta\mathbf{m}_i = [\Delta\hat{\mathbf{a}}_i, \Delta\hat{r}_i, \Delta\hat{m}_i]$ added after activation with bounded magnitude:
\begin{equation}
\mathbf{m}_i^{\mathrm{f}}=\mathrm{clamp}\bigl(\mathbf{m}_i+\boldsymbol{\sigma}(\Delta\hat{\mathbf{m}}_i)\odot\boldsymbol{\delta},\,\mathbf{m}_{\min},\,\mathbf{m}_{\max}\bigr),
\end{equation}
where $\boldsymbol{\delta}=[\delta_a,\delta_r,\delta_m]$ controls per-channel residual magnitude and clamping enforces physical ranges. The residuals are $\ell_2$-regularized to remain small, preserving the palette as the primary representation.

\paragraph{Initialization and merging.}
The palette is initialized by MiniBatch K-Means on SH DC coefficients, followed by MLP pre-training. During training, similar entries are periodically merged: pairwise distances are computed on activated palette parameters with albedo weighted higher than scalar BRDF attributes ($\mathbf{W}_{\mathrm{feat}}{=}[1,1,1,0.5,0.5]$); entries whose distance falls below a threshold ($\epsilon_{\mathrm{merge}}{=}0.08$) are grouped via Union-Find; each cluster is replaced by a weighted average in raw parameter space using Gaussian assignment weights; and the MLP is re-initialized with pre-training on the new assignments, allowing the method to discover the natural material count of the scene.

\subsection{Physically Based Rendering}

\paragraph{Material rasterization.}
Per-Gaussian materials are rasterized into image-space attributes via the same transmittance-based compositing used for Gaussian rendering:
\begin{equation}
\bar{\mathbf{m}}(\mathbf{x}) = \sum_{i=1}^{N}T_i\alpha_i\,\mathbf{m}_i^{\mathrm{f}},
\end{equation}
where $T_i=\prod_{j=1}^{i-1}(1-\alpha_j)$. We recover pixel-space attributes through alpha normalization to avoid $\alpha^2$ darkening at edges:
\begin{equation}
\mathbf{m}_{\mathrm{px}}=\frac{\bar{\mathbf{m}}}{\bar{\alpha}+\epsilon}.
\end{equation}

\paragraph{Surface normal estimation.}
Surface normals are derived from the rasterized depth map via finite-difference cross products:
\begin{equation}
\mathbf{n}(\mathbf{x}) = \frac{\partial \mathbf{P}}{\partial x} \times \frac{\partial \mathbf{P}}{\partial y},
\end{equation}
where $\mathbf{P}(\mathbf{x}) = \mathbf{c}_{\mathrm{cam}} + d(\mathbf{x}) \cdot \mathbf{r}(\mathbf{x})$ is the surface point recovered from depth. We prefer depth-derived normals over alpha-blended Gaussian normals, since alpha-blending commutes with diffuse shading but not with the nonlinear specular BRDF~\cite{glossygs}.

\paragraph{Direct illumination.}
We evaluate the rendering integral via Monte Carlo sampling with $S$ directions distributed on the upper hemisphere using Fibonacci sphere sampling~\cite{burley12disneybrdf}. The outgoing radiance is decomposed into diffuse and specular components. The diffuse term is:
\begin{equation}
L_{\mathrm{diffuse}} = \frac{\mathbf{a}_{\mathrm{px}}}{\pi}\int_{\Omega^+} L_i(\boldsymbol{\omega}_i)\,(\mathbf{n}\cdot\boldsymbol{\omega}_i)\,d\boldsymbol{\omega}_i.
\end{equation}
The specular term uses the GGX microfacet BRDF~\cite{walter2007ggx} with Schlick's Fresnel approximation~\cite{schlick1994brdf}:
\begin{equation}
f_s(\mathbf{p},\boldsymbol{\omega}_i,\boldsymbol{\omega}_o) = \frac{F(\boldsymbol{\omega}_o,\mathbf{h})\cdot D(\mathbf{h})\cdot G(\boldsymbol{\omega}_i,\boldsymbol{\omega}_o)}{4\,(\mathbf{n}\cdot\boldsymbol{\omega}_i)\,(\mathbf{n}\cdot\boldsymbol{\omega}_o)},
\end{equation}
where $\mathbf{h}$ is the half-vector, $\alpha = r_{\mathrm{px}}^2$, and the base reflectivity interpolates between dielectric ($F_0 = 0.04$) and metallic ($F_0 = \mathbf{a}_{\mathrm{px}}$) behavior via $F_0 = (1 - m_{\mathrm{px}}) \cdot 0.04 + m_{\mathrm{px}} \cdot \mathbf{a}_{\mathrm{px}}$.

\paragraph{Indirect illumination.}
For forward relighting, we employ a wavefront path tracer that proceeds from the second bounce with cosine-weighted hemisphere sampling, next-event estimation, and Russian roulette termination. During inverse rendering, we approximate indirect illumination to avoid the cost of differentiating through multi-bounce path tracing (e.g., via path replay backpropagation~\cite{PathReplayBackpropagation}). The incident radiance at each shading point is
\begin{equation}
L_i(\mathbf{p},\boldsymbol{\omega}_i) = V(\mathbf{p},\boldsymbol{\omega}_i)\,L_{\mathrm{env}}(\boldsymbol{\omega}_i) + L_{\mathrm{indirect}}(\mathbf{p},\boldsymbol{\omega}_i),
\end{equation}
where $V$ is the binary visibility from ray tracing through a Gaussian surfel BVH, $L_{\mathrm{env}}$ is the environment map, and the indirect term is read from the pre-trained SH coefficients of the intersected Gaussian:
\begin{equation}
L_{\mathrm{indirect}}(\mathbf{p},\boldsymbol{\omega}_i) = (1 - V(\mathbf{p},\boldsymbol{\omega}_i)) \cdot C_{\mathrm{trace}}(\mathbf{p},\boldsymbol{\omega}_i).
\end{equation}
Each Gaussian surfel is converted to a bounding icosahedron mesh based on its covariance and opacity threshold for BVH construction. Following IRGS~\cite{IRGS}, the pre-trained SH radiance field serves as a compact indirect illumination cache during optimization. Once materials are decomposed and edited, the SH cache is replaced by forward path tracing for relighting.

\paragraph{Environment map.}
Illumination is represented as a learnable lat-long HDR environment map $L_{\mathrm{env}} \in \mathbb{R}^{H \times 2H \times 3}$ with exponential activation. A mipmap pyramid provides pre-filtered specular convolution indexed by roughness, and importance sampling combines cosine-weighted hemisphere and environment map sampling.

\subsection{Residual Weight Estimation}

The PBR model cannot fully explain all appearance effects in the pre-trained scene, particularly where reconstructed geometry is inaccurate. To identify which Gaussians need additional correction, we train a per-Gaussian residual weight simultaneously with an SH residual field that absorbs appearance the PBR model cannot capture. The rendered image blends PBR and SH outputs:
\begin{equation}
\hat{C}_{\mathrm{blend}}(\mathbf{x})=(1-\bar{w}(\mathbf{x}))\hat{C}_{\mathrm{PBR}}(\mathbf{x})+\bar{w}(\mathbf{x})\hat{C}_{\mathrm{SH}}(\mathbf{x}),
\end{equation}
where $\bar{w}(\mathbf{x})$ is the rasterized residual weight and each Gaussian has $w_i = \sigma(\hat{w}_i) \cdot w_{\max}$. The primary output of this stage is the learned weights $w_i$ themselves, not the blended image; the SH residual field is discarded after this stage.

Not all Gaussians need correction: a top-$k$ sparsity mask with straight-through gradients keeps only the top $\rho$ fraction non-zero in the forward pass while allowing gradients to flow through all Gaussians, with $\rho$ ramped from 0 during training. At inference time, $\rho$ can be manually reduced to trade reconstruction fidelity for stronger palette editability.

The residual weights are further regularized to match the pre-computed error distribution:
\begin{equation}
\mathcal{L}_{\mathrm{rw}} = \left\| \bar{w}(\mathbf{x}) - \mathrm{clamp}(s \cdot (e(\mathbf{x}) - \epsilon_e),\, 0,\, w_{\max}) \right\|^2,
\end{equation}
where $e(\mathbf{x})$ is the per-pixel reconstruction error from the pure PBR render, encouraging large weights only where decomposition error is significant.

\subsection{Per-Gaussian Material Refinement}

The palette-based representation guarantees global albedo uniformity but may lack expressiveness for fine local material variation. After the residual weight has been learned, we freeze it and introduce per-Gaussian direct materials $\mathbf{d}_i = [\hat{\mathbf{a}}_i^{\mathrm{direct}}, \hat{r}_i^{\mathrm{direct}}, \hat{m}_i^{\mathrm{direct}}]$. The final material is a weight-guided blend:
\begin{equation}
\mathbf{m}_i^{\mathrm{final}} = (1 - w_i) \cdot \mathbf{m}_i^{\mathrm{palette}} + w_i \cdot \sigma(\hat{\mathbf{d}}_i),
\end{equation}
where $\sigma(\hat{\mathbf{d}}_i)$ denotes element-wise activation.
Since $w_i$ is large only where the palette decomposition is inaccurate, per-Gaussian detail is added exactly where needed. Direct materials are initialized from the current palette output; all other parameters are frozen in this stage.

\subsection{Training Objectives}

Our loss function balances physical plausibility, compactness, and fidelity. The main reconstruction objective is
\begin{equation}
\mathcal{L}_{\mathrm{recon}}=(1-\lambda_{\mathrm{DSSIM}})\,\mathcal{L}_1+\lambda_{\mathrm{DSSIM}}\bigl(1-\operatorname{SSIM}\bigr).
\end{equation}

\paragraph{Normal consistency.}
We enforce consistency between rendered normals and depth-derived pseudo-normals:
\begin{equation}
\mathcal{L}_{\mathrm{normal}} = \frac{1}{|\mathcal{M}|}\sum_{\mathbf{x}\in\mathcal{M}}\bigl(1 - \mathbf{n}_{\mathrm{render}}(\mathbf{x})\cdot\mathbf{n}_{\mathrm{depth}}(\mathbf{x})\bigr).
\end{equation}

\paragraph{Truncated albedo consistency.}
An edge-aware loss that encourages smooth albedo within material regions while preserving boundaries:
\begin{equation}
\mathcal{L}_{\mathrm{albedo}}=\frac{1}{C}\sum_{c=1}^{C}\frac{1}{HW}\sum_{\mathbf{x}}\min\left(\left|a_c(\mathbf{x})-\bar{a}_c^{\mathrm{nb}}(\mathbf{x})\right|,\epsilon_a\right)g(\mathbf{x}),
\end{equation}
where $\bar{a}_c^{\mathrm{nb}}$ is the $3\times3$ neighborhood mean, $\epsilon_a$ is the truncation threshold, and $g(\mathbf{x})=\exp(-5\|\nabla C_{\mathrm{GT}}(\mathbf{x})\|)$ suppresses smoothing at GT image edges. This loss is used only in Stage~2 with its weight annealed from $\lambda{=}0.5$ to $0$, and disabled in Stage~4 since per-Gaussian materials are free to deviate from the palette.

\paragraph{Palette regularization.}
To prevent palette collapse, we apply an entropy loss encouraging diverse usage:
\begin{equation}
\mathcal{L}_{\mathrm{entropy}} = \sum_{k=1}^{K}\bar{w}_k\log\bar{w}_k, \quad \bar{w}_k = \frac{1}{N}\sum_{i=1}^{N}w_{ik},
\end{equation}
and a dead palette penalty that softly pushes the optimizer to use or merge unused entries:
\begin{equation}
\mathcal{L}_{\mathrm{dead}} = \sum_{k=1}^{K}\max\bigl(0,\, \epsilon_{\mathrm{dead}} - \bar{w}_k\bigr)^2.
\end{equation}

\paragraph{Additional terms.}
A lighting smoothness loss $\mathcal{L}_{\mathrm{light}}=\|L_{\mathrm{direct}}-\bar{L}_{\mathrm{direct}}\|_1$ encourages spatially smooth direct lighting and is used only in Stage~2. The material residuals are regularized with an $\ell_2$ penalty to remain small.

%% file: chapters/results.tex
\section{RESULTS}

We evaluate the proposed method on decomposition quality, rendering fidelity, and editing capability.

\subsection{Experimental Setup}

The optimization proceeds in four stages starting from a pre-trained Ref-GS~\cite{zhang2024ref} scene (Stage~1). Stage~2 jointly optimizes the material palette, spatial field, and environment map for 3000 iterations (${\sim}20$ min). Stage~3 jointly trains per-Gaussian residual weights and the SH radiance field for 1000 iterations (${\sim}2.5$ min) to identify decomposition errors. Stage~4 freezes the learned weights and refines per-Gaussian direct materials for 3000 iterations (${\sim}17$ min). The total optimization time (excluding Stage~1) is approximately 40 minutes. All experiments use a sparsity ratio of $\rho{=}16\%$, which provides sufficient capacity for most objects while keeping the refinement stage compact.

\subsection{Comparisons}

We evaluate on two synthetic benchmarks with ground-truth materials and relighting images: TensoIR~\cite{Jin2023TensoIR} (Table~\ref{tab:comparison_tensoir}) and Synthetic4Relight~\cite{zhang2022invrender} (Table~\ref{tab:comparison_s4r}). Results show that our method achieves the best albedo decomposition on both datasets while maintaining competitive relighting quality. The palette representation naturally suppresses the high-frequency noise typical of per-primitive estimation and enforces global material consistency, which is especially beneficial for albedo recovery. Our relighting PSNR also ranks close to the strongest baselines (IRGS, SVG-IR), which we attribute to the high quality of the recovered materials---when the decomposed albedo, roughness, and environment map are physically accurate, the relit images are faithful regardless of the rendering formulation.

We further compare palette-based editing against PaletteNeRF~\cite{kuang2023palettenerf} and PaletteGaussian~\cite{ren2024palettegaussian} on the Lego scene using FLIP~\cite{Andersson2020} (Figure~\ref{fig:palette_compare}). RGB-only palette methods exhibit editing leakage into semantically distinct regions, while our BRDF-level palette confines edits to the intended material. Figure~\ref{fig:edit_relight} shows relighting results after editing the roughness and metallic of a single palette entry across multiple scenes. The edited materials produce view-dependent specular highlights that vary naturally across viewpoints and environment maps, confirming that the recovered BRDF parameters are physically meaningful and that palette-level edits propagate consistently under novel illumination. We provide an SH vs.\ path tracing comparison in the supplementary material.

\begin{table}[t]
  \centering
  \small
  \caption{Comparison on the TensoIR dataset. Our method achieves the best albedo decomposition while maintaining competitive relighting quality.}
  \label{tab:comparison_tensoir}
  \setlength{\tabcolsep}{3pt}
  \begin{tabularx}{\linewidth}{l*{5}{>{\centering\arraybackslash}X}}
    \toprule
    & Relight & \multicolumn{3}{c}{Albedo} & Rough. \\
    \cmidrule(lr){2-2}\cmidrule(lr){3-5}\cmidrule(lr){6-6}
    Method & PSNR$\uparrow$ & PSNR$\uparrow$ & SSIM$\uparrow$ & LPIPS$\downarrow$ & MSE$\downarrow$ \\
    \midrule
    TensoIR~\cite{Jin2023TensoIR} & 28.58 & 29.28 & \cellcolor{red!15}0.950 & 0.085 & \cellcolor{red!50}0.013 \\
    GS-IR~\cite{gsir2023} & 24.37 & \cellcolor{red!15}30.29 & 0.941 & 0.084 & 0.027 \\
    R3DG~\cite{R3DG} & 24.41 & 29.47 & 0.930 & 0.107 & \cellcolor{red!30}0.016 \\
    IRGS~\cite{IRGS} & \cellcolor{red!30}29.91 & \cellcolor{red!30}33.80 & \cellcolor{red!30}0.954 & \cellcolor{red!30}0.074 & 0.055 \\
    SVG-IR~\cite{sun25svgir} & \cellcolor{red!50}30.29 & 30.13 & 0.947 & \cellcolor{red!15}0.076 & 0.033 \\
    RadioGS~\cite{han2026radiogs} & 26.00 & 27.62 & 0.937 & 0.085 & \cellcolor{red!15}0.025 \\
    \midrule
    Ours & \cellcolor{red!15}28.90 & \cellcolor{red!50}34.16 & \cellcolor{red!50}0.960 & \cellcolor{red!50}0.067 & 0.043 \\
    \bottomrule
  \end{tabularx}
\end{table}

\begin{table}[t]
  \centering
  \small
  \caption{Comparison on the Synthetic4Relight dataset. Our method achieves the best albedo PSNR and roughness estimation with competitive relighting quality, demonstrating strong cross-dataset generalization.}
  \label{tab:comparison_s4r}
  \setlength{\tabcolsep}{3pt}
  \begin{tabularx}{\linewidth}{l*{5}{>{\centering\arraybackslash}X}}
    \toprule
    & Relight & \multicolumn{3}{c}{Albedo} & Rough. \\
    \cmidrule(lr){2-2}\cmidrule(lr){3-5}\cmidrule(lr){6-6}
    Method & PSNR$\uparrow$ & PSNR$\uparrow$ & SSIM$\uparrow$ & LPIPS$\downarrow$ & MSE$\downarrow$ \\
    \midrule
    GS-IR~\cite{gsir2023} & 25.40 & 19.48 & 0.896 & 0.117 & 0.011 \\
    R3DG~\cite{R3DG} & 31.00 & \cellcolor{red!15}28.31 & \cellcolor{red!30}0.951 & \cellcolor{red!30}0.058 & 0.013 \\
    IRGS~\cite{IRGS} & \cellcolor{red!50}34.90 & \cellcolor{red!30}30.81 & \cellcolor{red!50}0.957 & \cellcolor{red!50}0.055 & \cellcolor{red!30}0.008 \\
    SVG-IR~\cite{sun25svgir} & \cellcolor{red!15}31.55 & 28.29 & 0.937 & 0.069 & 0.020 \\
    RadioGS~\cite{han2026radiogs} & 29.81 & 27.03 & \cellcolor{red!15}0.949 & \cellcolor{red!15}0.063 & \cellcolor{red!15}0.010 \\
    \midrule
    Ours & \cellcolor{red!30}32.45 & \cellcolor{red!50}31.88 & 0.946 & 0.083 & \cellcolor{red!50}0.007 \\
    \bottomrule
  \end{tabularx}
\end{table}

\subsection{Ablation Studies}

We ablate the key design choices: initial palette size ($K$), palette merging (+Merge), and per-material refinement (+Refine). Table~\ref{tab:ablation} reports rendering quality (NVS PSNR, SSIM, LPIPS) and material decomposition quality (Albedo PSNR, Roughness MSE) on the TensoIR dataset.

\begin{table}[t]
  \centering
  \small
  \caption{Ablation study. $K$ denotes the initial palette count. +Merge enables palette merging, and +Refine enables per-material refinement.}
  \label{tab:ablation}
  \setlength{\tabcolsep}{4.5pt}
  \renewcommand{\arraystretch}{1.08}
  \begin{tabular}{
    @{}l
    S[table-format=2.2]
    S[table-format=2.2]
    S[table-format=1.4]
    S[table-format=1.4]
    S[table-format=1.4]
    @{}
  }
    \toprule
    \multirow{2}{*}{Configuration}
    & \multicolumn{1}{c}{NVS}
    & \multicolumn{3}{c}{Albedo}
    & \multicolumn{1}{c}{Roughness} \\
    \cmidrule(lr){2-2}
    \cmidrule(lr){3-5}
    \cmidrule(lr){6-6}
    & {PSNR$\uparrow$}
    & {PSNR$\uparrow$}
    & {SSIM$\uparrow$}
    & {LPIPS$\downarrow$}
    & {MSE$\downarrow$} \\
    \midrule
    $K{=}4$                  & 29.95 & 31.94 & 0.9475 & 0.0852 & 0.0302 \\
    $K{=}8$                  & 29.65 & 32.23 & 0.9490 & 0.0827 & 0.0312 \\
    $K{=}24$ + Merge          & 31.22 & 31.90 & 0.9487 & 0.0721 & 0.0347 \\
    $K{=}8$ + Merge           & 30.40 & 32.37 & 0.9492 & 0.0827 & 0.0305 \\
    \midrule
    $K{=}8$ + Merge + Refine   & 32.19 & 32.90 & 0.9506 & 0.0789 & 0.0285 \\
    \bottomrule
  \end{tabular}
\end{table}

\paragraph{Effect of initial palette size ($K$).}
Starting with $K{=}24$ and merging yields higher NVS PSNR than $K{=}8$ without merging but lower than $K{=}8$ +Merge. A larger initial palette provides more capacity but the merging process must work harder to collapse similar entries; moreover, with too many entries the optimization tends to bake lighting into separate palette materials instead of factoring it into the environment map, producing entries that are visually similar but numerically distinct and thus difficult to merge. $K{=}4$ produces low albedo quality, indicating that too few palette entries cannot adequately represent the material diversity of the scene. Moreover, higher albedo or roughness metrics do not necessarily imply a better palette assignment: when $K$ is too small, visually distinct materials are forced to share palette entries, making subsequent palette-level editing difficult or impossible.

\paragraph{Effect of palette merging (+Merge).}
Comparing $K{=}8$ +Merge against $K{=}8$ without merging, merging improves both NVS and Albedo PSNR. Merging allows the optimization to discover the natural material count of the scene by removing redundant palette entries, which reduces over-segmentation and concentrates representational capacity on genuinely distinct materials.

\paragraph{Effect of per-material refinement (+Refine).}
Comparing $K{=}8$ +Merge +Refine against $K{=}8$ +Merge, the refinement stage consistently improves both NVS and Albedo PSNR. The improvement is moderate rather than drastic, and the required sparsity ratio is well below the 16\% used in our experiments, confirming that only a small fraction of Gaussians deviate from the palette. This preserves the strong editability of the palette: the vast majority of Gaussians remain palette-driven, and per-Gaussian corrections are confined to localized regions where the PBR decomposition falls short.

%% file: chapters/discussion.tex
\section{DISCUSSION}

% The main benefit of the palette design is that it changes how errors are allowed to enter the decomposition. In per-Gaussian inverse rendering, each primitive can absorb errors independently. In our method, a material prototype must explain many Gaussians at once, so spatially varying shadows or visibility errors are less likely to be mistaken for material variation.
% The palette-based design introduces a structural prior that offers three main advantages. First, global material consistency: by sharing palette entries across Gaussians, the decomposition naturally avoids the high-frequency noise typical of per-primitive estimation and acts as a spatial regularizer that prevents residual lighting from being baked into albedo---since all Gaussians assigned to the same entry must share identical material parameters, the optimization is forced to attribute spatially varying shading to illumination rather than albedo. Second, compactness: reducing the material space from per-primitive parameters to a small palette plus a spatial field concentrates capacity where it matters most, and editing one entry immediately affects all associated regions. Third, the continuous spatial material field ties assignments to geometry through position, yielding robust material predictions even in sparsely observed regions.

The main advantage of the palette-based formulation is that it changes how reconstruction errors are explained during inverse rendering. In conventional per-Gaussian optimization, each primitive can independently absorb shading residuals, causing shadows, visibility errors, or indirect illumination to be baked into local material estimates. In contrast, our method constrains many Gaussians to share the same BRDF prototype, making spatially varying effects harder to explain through independent material fitting and encouraging the optimization to attribute them to illumination instead of albedo.

This structural prior brings three benefits. First, shared palette entries enforce \emph{global material consistency} and suppress the high-frequency noise typical of per-primitive estimation. Second, the compact palette representation greatly \emph{reduces the dimensionality of the material space} while enabling intuitive editing, since modifying one palette entry consistently affects all associated regions. Third, the continuous spatial material field links material assignments to geometry, \emph{improving robustness in sparsely observed regions}.

At the same time, the method does not fully resolve the tension between physical interpretability and visual fidelity. While the PBR model captures dominant reflectance behavior, the pre-trained Gaussian appearance may still contain residual effects that cannot be explained by a compact BRDF model. The SH-blended training stage absorbs part of these discrepancies, but excessive residual compensation may require stronger per-Gaussian refinement, weakening the regularization effect of the palette.

Another limitation is the difficulty of representing complex spatially-varying textures. The palette assumes that materials can be described by a small set of shared prototypes, which works well for piecewise-uniform materials but struggles with high-frequency appearance variation such as wood grain, fabric weave, or weathered surfaces. In real-world scenes where material appearance is highly diverse and fine-grained, the assumption that materials cluster into a compact set breaks down, and the palette may either over-segment the scene or produce overly smooth decompositions that lose important detail.

Finally, our framework depends on the geometric fidelity of the pre-trained Ref-GS~\cite{zhang2024ref} representation. Errors in normals, visibility, or surface coverage can directly affect decomposition quality, suggesting future directions in joint geometry-material optimization, richer BRDF models, and hierarchical or uncertainty-aware palette representations.

%% file: chapters/conclusion.tex
\section{CONCLUSION}

We present a palette-based inverse rendering framework for 2D Gaussian Splatting that enables compact, editable, and physically based material decomposition. Instead of assigning unconstrained BRDF parameters to individual Gaussians, our method represents scene materials using a shared palette of BRDF prototypes predicted through a continuous spatial material field. Combined with physically based rendering, residual weight estimation, and lightweight per-Gaussian refinement, the framework produces unsupervised material decompositions that remain both visually faithful and physically coherent without requiring ground-truth material annotations. The resulting representation supports intuitive material editing, including palette-level appearance modification, cross-scene material transfer, and relighting under novel illumination.

The broader insight is that the main bottleneck in Gaussian inverse rendering is not only how accurately we render, but also how freely we allow materials to vary. Reducing unnecessary local material freedom makes the recovered attributes more coherent, more editable, and less likely to absorb baked illumination. 

Several directions remain open for future work, including tighter joint optimization of geometry and materials, richer BRDF parameterizations for complex layered surfaces, semantically driven material editing that leverages language models to associate palette entries with meaningful material categories, and cross-scene palette sharing for improved scalability and controllability.

%% file: chapters/images.tex
\clearpage

\begin{figure*}[t]
  \centering
  \includegraphics[width=\textwidth]{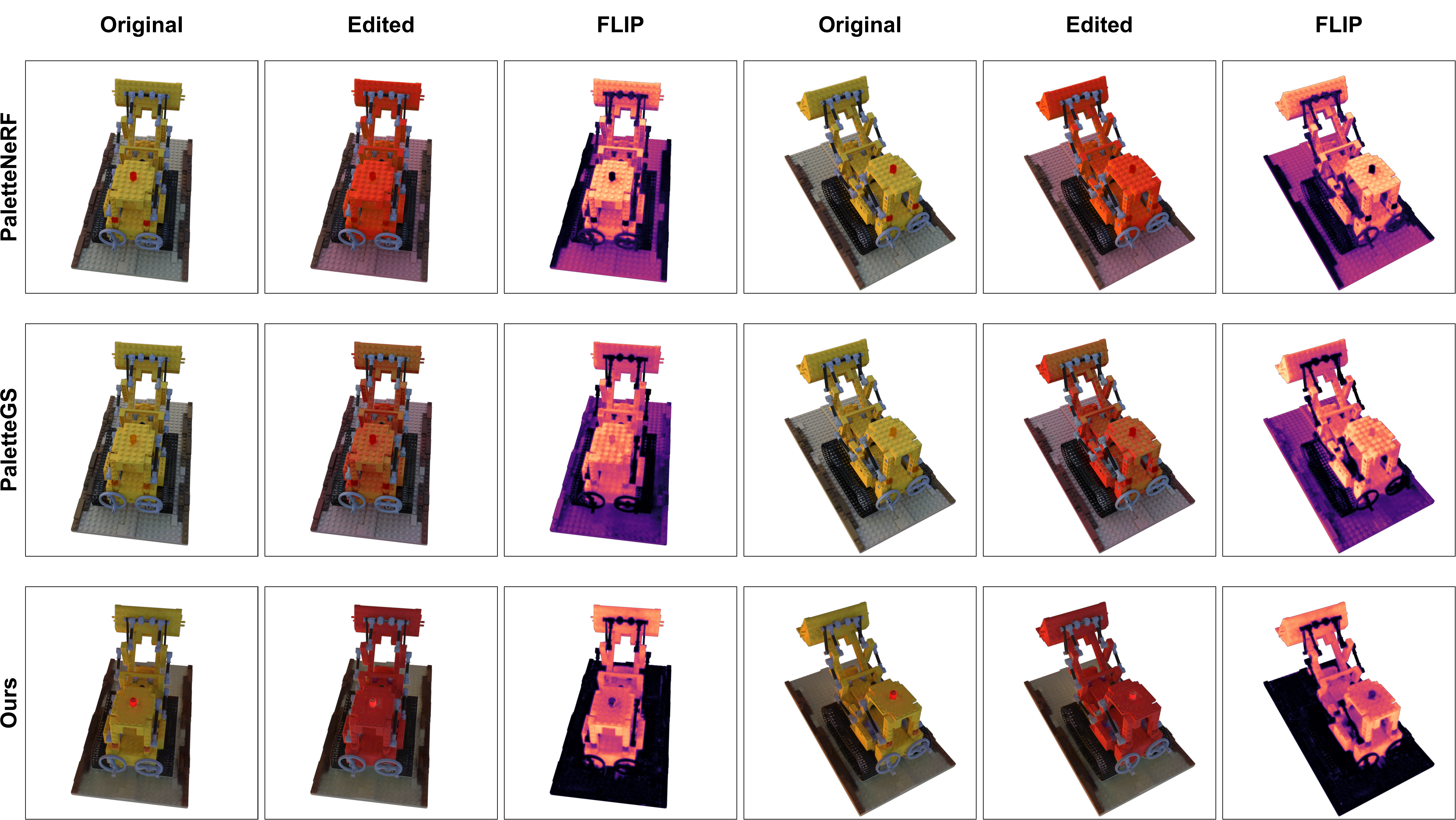}
  \captionof{figure}{Palette-based editing comparison on the Lego scene with PaletteNeRF~\cite{kuang2023palettenerf} and PaletteGaussian~\cite{ren2024palettegaussian}, evaluated using FLIP~\cite{Andersson2020}. RGB-only palette methods exhibit visible editing leakage, while our BRDF-level palette confines edits to the intended material region.}
  \label{fig:palette_compare}

  \includegraphics[width=\textwidth]{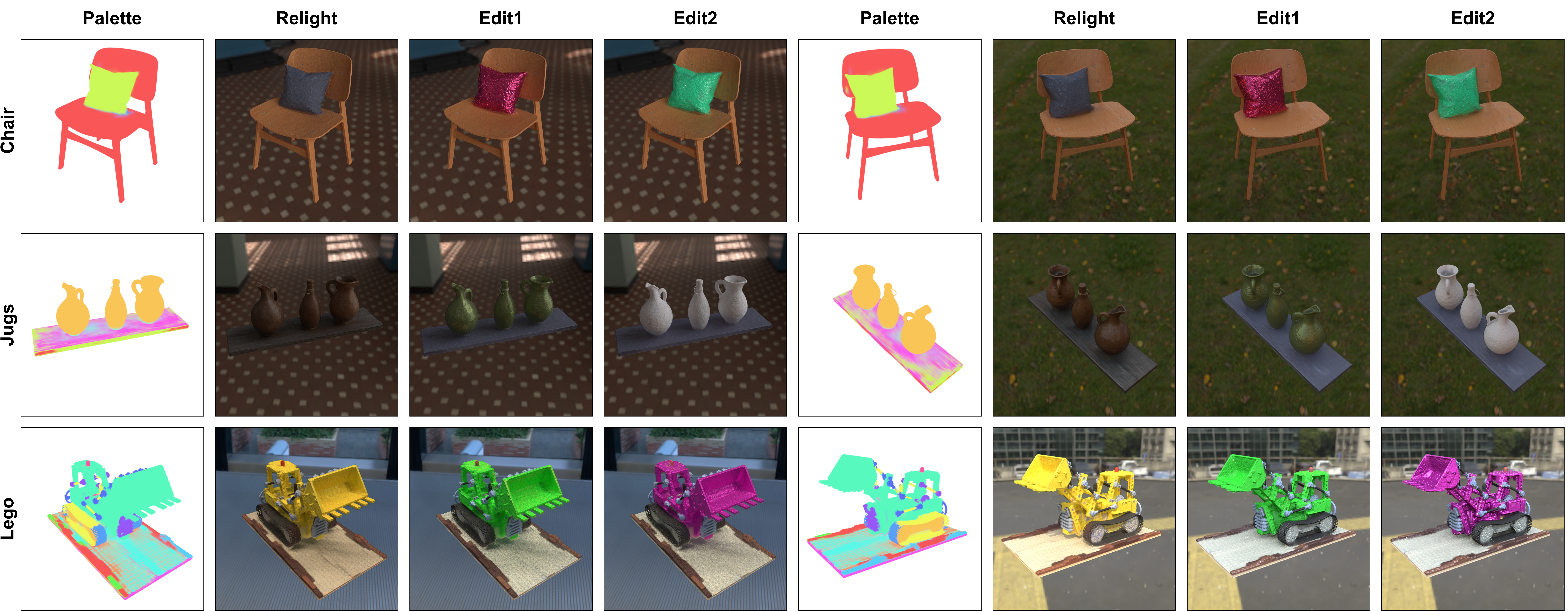}
  \captionof{figure}{Relighting results after palette-level material editing across multiple scenes and environment maps. We edit roughness and metallic of a single palette entry, producing view-dependent specular highlights that vary naturally across viewpoints and illumination---confirming that the recovered materials are physically meaningful rather than baked appearance.}
  \label{fig:edit_relight}
\end{figure*}

\begin{figure*}[t]
  \centering
  \includegraphics[width=\textwidth]{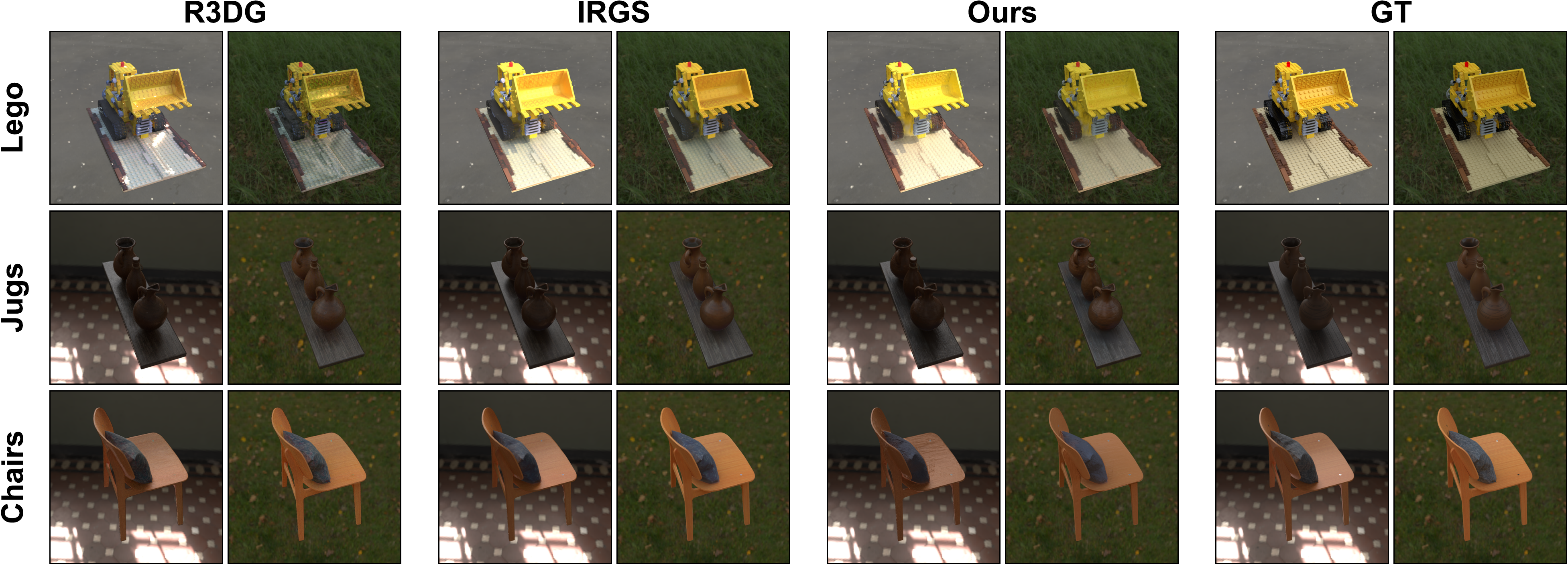}
  \caption{Relighting comparison. Our method produces physically plausible relighting results under novel environment maps, comparable to recent inverse rendering baselines. The palette-based decomposition recovers accurate materials and environment maps, enabling faithful relighting without baking illumination into material parameters.}
  \label{fig:relight_comparison}
\end{figure*}

\begin{figure*}[t]
  \centering
  \begin{subfigure}[b]{0.49\textwidth}
    \includegraphics[width=\textwidth]{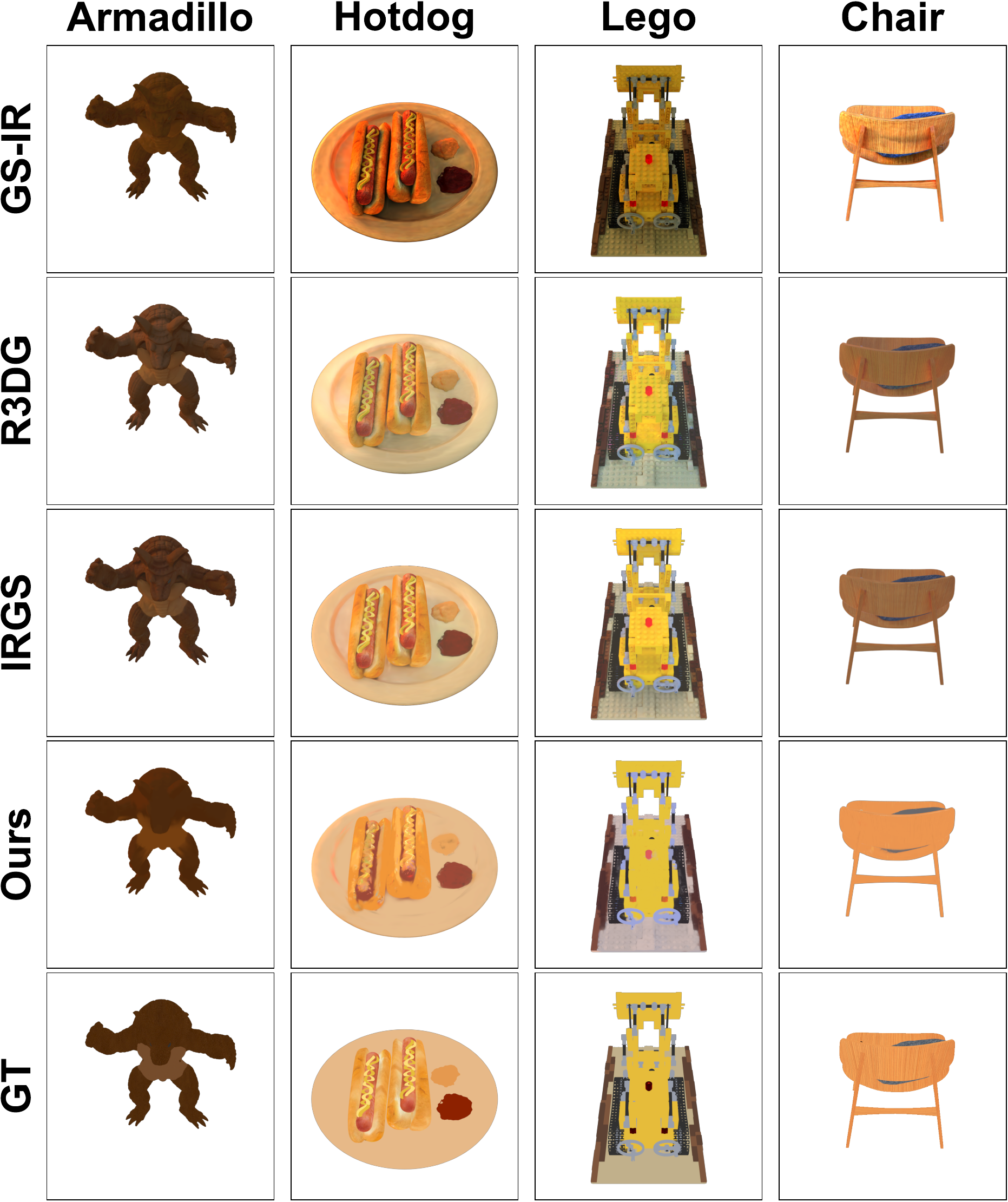}
    \caption{Albedo}
    \label{fig:albedo_comparison}
  \end{subfigure}
  \hfill
  \begin{subfigure}[b]{0.49\textwidth}
    \includegraphics[width=\textwidth]{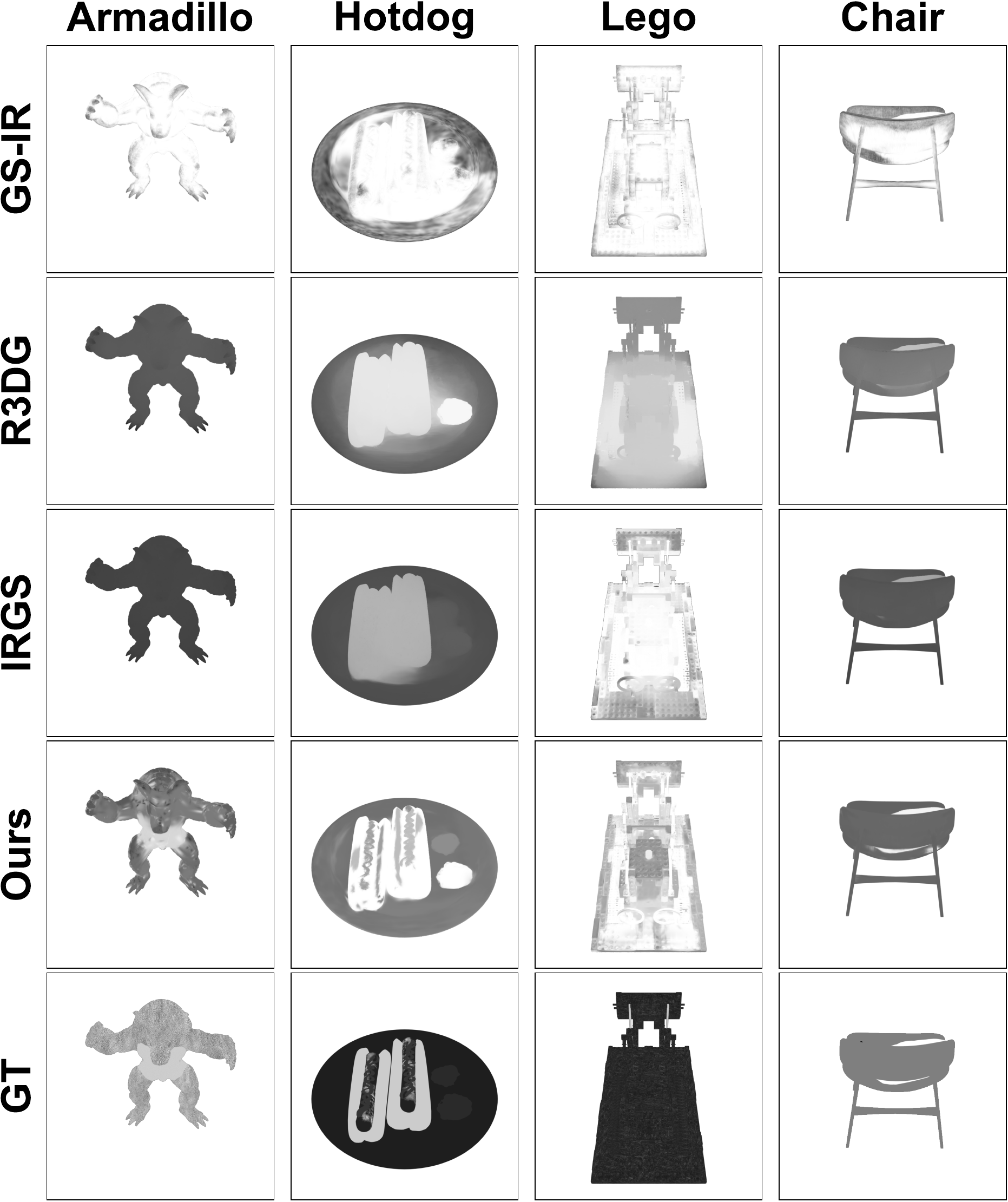}
    \caption{Roughness}
    \label{fig:roughness_comparison}
  \end{subfigure}
  \caption{Material decomposition comparison. (a) Albedo and (b) roughness maps recovered by each method. By sharing BRDF prototypes across Gaussians, our palette-based approach suppresses the high-frequency noise and spatial inconsistency typical of per-primitive estimation, producing cleaner and more globally coherent decompositions.}
  \label{fig:material_comparison}
\end{figure*}

%% file: chapters/supplement.tex
\clearpage

\appendix

\section{Appendix}

This appendix collects compact mathematical material that supports the main method section. It is written as a concise technical summary rather than a full derivation chapter, making it suitable for supplementary details in the chapter-based manuscript structure.

\subsection{Implementation Details}

Table~\ref{tab:hyperparameters} lists the default hyperparameters used across all experiments.

\begin{table}[b]
  \centering
  \small
  \caption{Default hyperparameters.}
  \label{tab:hyperparameters}
  \begin{tabularx}{\linewidth}{lXr}
    \toprule
    Parameter & Symbol & Default \\
    \midrule
    Number of palettes & $K$ & 8 \\
    MLP hidden dim / layers & $H$ / $D$ & 64 / 3 \\
    PE frequencies & $L$ & 6 \\
    Temperature & $\tau_0 \to \tau_{\min}$ & 0.1 $\to$ 0.01 \\
    Albedo threshold & $\epsilon_a$ & 0.15 \\
    Envmap resolution & --- & 32 \\
    Albedo scale & $a_s$ & 0.94 \\
    Albedo bias & $a_b$ & 0.03 \\
    Residual blend max & $w_{\max}$ & 0.8 \\
    \bottomrule
  \end{tabularx}
\end{table}

\subsection{Notation}

Table~\ref{tab:notation} summarizes the primary symbols used throughout the manuscript.

\begin{table}[b]
  \centering
  \caption{Main notation used in the paper.}
  \label{tab:notation}
  \footnotesize
  \setlength{\tabcolsep}{4pt}
  \begin{tabular}{>{\raggedright\arraybackslash}p{0.09\linewidth}>{\raggedright\arraybackslash}p{0.33\linewidth}>{\raggedright\arraybackslash}p{0.09\linewidth}>{\raggedright\arraybackslash}p{0.33\linewidth}}
    \toprule
    Symbol & Meaning & Symbol & Meaning \\
    \midrule
    $N$ & Number of Gaussians & $K$ & Number of material palette entries \\
    $\mathbf{m}_k$ & $k$-th palette material & $\mathcal{F}_{\theta}$ & Spatial material field \\
    $\tau$ & Softmax temperature & $w_{ik}$ & Assignment weight \\
    $\alpha_i$ & Gaussian opacity & $T_i$ & Transmittance \\
    $\mathbf{n}$ & Surface normal & $\boldsymbol{\omega}_o$ & Outgoing direction \\
    $\boldsymbol{\omega}_i$ & Incoming direction & $\mathbf{h}$ & Half-vector \\
    $f_s$ & Specular BRDF & $L_{\mathrm{env}}$ & Environment radiance \\
    $V$ & Visibility & $S$ & Number of light samples \\
    $w_i$ & Residual weight & $\epsilon_a$ & Truncation threshold \\
    $\epsilon_{\mathrm{merge}}$ & Merge threshold & $F_0$ & Base reflectivity \\
    $\alpha$ & GGX roughness squared & $\Omega_j$ & Monte Carlo sample weight \\
    \bottomrule
  \end{tabular}
\end{table}

\subsection{Theoretical Analysis}

\paragraph{Albedo uniformity.}
Under the palette-based representation with sufficiently low temperature $\tau$, the albedo of all Gaussians assigned to the same palette entry is identical by construction. Let Gaussian $i$ be assigned to palette $k^*$, i.e.\ $w_{ik^*}\to 1$ as $\tau\to 0$. Then
\begin{equation}
\mathbf{a}_i=\sum_{k=1}^{K}w_{ik}\mathbf{a}_k\xrightarrow{\tau\to 0}\mathbf{a}_{k^*}.
\end{equation}
Since $\mathbf{a}_{k^*}$ is a shared parameter, all Gaussians assigned to palette $k^*$ share the same albedo. This is a structural advantage over per-Gaussian representations, where even with smoothness regularization, nearby Gaussians can drift to slightly different albedo values, causing visible seams and non-physical texture.

\paragraph{Expressiveness of soft blending.}
While the dominant palette assignment is near-one-hot for low $\tau$, the soft blending provides smooth transitions at material boundaries and allows sub-material variation through differing blend ratios. The effective per-pixel material space is the convex hull of the $K$ palette entries:
\begin{equation}
\mathcal{M}_{\mathrm{eff}}=\left\{\sum_{k=1}^{K}w_k\mathbf{m}_k\;\middle|\;w_k\geq 0,\;\sum_{k=1}^{K}w_k=1\right\}\subseteq\mathbb{R}^{5},
\end{equation}
a $(K{-}1)$-dimensional simplex in material space that is richer than $K$ discrete materials while remaining globally consistent.

\paragraph{Spatial coherence via Lipschitz continuity.}
We now show that the spatial material field is Lipschitz continuous and derive an explicit bound on the material variation between nearby positions.

\begin{proposition}
Let $\mathbf{w}(\mathbf{x})=\operatorname{softmax}(\mathcal{F}_\theta(\gamma(\bar{\mathbf{x}}))/\tau)$ be the palette assignment field. Then for any two positions $\mathbf{x}_1,\mathbf{x}_2$,
\begin{equation}
\lVert\mathbf{w}(\mathbf{x}_1)-\mathbf{w}(\mathbf{x}_2)\rVert_2\leq L_{\mathcal{F}}\,\lVert\mathbf{x}_1-\mathbf{x}_2\rVert_2,
\end{equation}
where $L_{\mathcal{F}}=L_{\mathrm{sm}}\cdot L_{\mathrm{MLP}}\cdot L_\gamma\cdot L_{\mathrm{norm}}$ decomposes into the product of Lipschitz constants of each stage.
\end{proposition}

\begin{proof}
The field is a composition of four mappings:
\begin{equation}
\mathbf{x}\xrightarrow{g_1}\bar{\mathbf{x}}\xrightarrow{g_2}\gamma(\bar{\mathbf{x}})\xrightarrow{g_3}\mathcal{F}_\theta(\gamma(\bar{\mathbf{x}}))\xrightarrow{g_4}\mathbf{w}(\mathbf{x}).
\end{equation}
We bound the Lipschitz constant of each stage.

\textbf{Stage 1: Coordinate normalization.}
The mapping $g_1\colon\mathbf{x}\mapsto\bar{\mathbf{x}}$ normalizes each coordinate to $[0,1]$ via $\bar{x}_i=(x_i-x_{\min,i})/\Delta_i$ where $\Delta_i=x_{\max,i}-x_{\min,i}$. The Jacobian is diagonal with entries $1/\Delta_i$, so
\begin{equation}
L_{\mathrm{norm}}=\max_{i}\frac{1}{\Delta_i}.
\end{equation}

\textbf{Stage 2: Positional encoding.}
The encoding $\gamma$ maps $\bar{\mathbf{x}}\in\mathbb{R}^3$ to $\mathbb{R}^{3+6L}$. Each input dimension contributes independently to disjoint output dimensions. For input component $x_i$, the output derivatives are: the identity ($1$), $\sin$ at level $l$ with derivative $2^l\pi\cos(2^l\pi x_i)$, and $\cos$ at level $l$ with derivative $2^l\pi\sin(2^l\pi x_i)$. The squared column norm of the Jacobian for dimension $i$ is
\begin{equation}
\lVert\mathbf{j}_i\rVert^2=1+\sum_{l=0}^{L-1}\bigl[(2^l\pi)^2\cos^2(\cdot)+(2^l\pi)^2\sin^2(\cdot)\bigr]=1+\pi^2\sum_{l=0}^{L-1}4^l.
\end{equation}
Since the Jacobian columns have non-overlapping support, $J_\gamma^T J_\gamma$ is diagonal with $\lVert J_\gamma\rVert_2=\max_i\lVert\mathbf{j}_i\rVert$. Evaluating the geometric sum,
\begin{equation}
L_\gamma=\sqrt{1+\frac{\pi^2(4^L-1)}{3}}.
\end{equation}
Note that $L_\gamma$ grows exponentially with the number of frequency bands $L$, reflecting the high-frequency sensitivity of the encoding.

\textbf{Stage 3: MLP.}
For a ReLU network $\mathcal{F}_\theta$ with layer weight matrices $W_1,\ldots,W_D$, ReLU is $1$-Lipschitz, so by submultiplicativity of the spectral norm,
\begin{equation}
L_{\mathrm{MLP}}\leq\prod_{l=1}^{D}\lVert W_l\rVert_2.
\end{equation}
In practice, the spectral norms of the weight matrices are moderate (typically $O(1)$ per layer), and gradient-based training naturally keeps $L_{\mathrm{MLP}}$ bounded.

\textbf{Stage 4: Softmax with temperature.}
The softmax Jacobian is $J_{\mathrm{sm}}=\frac{1}{\tau}(\operatorname{diag}(\boldsymbol{\sigma})-\boldsymbol{\sigma}\boldsymbol{\sigma}^T)$, where $\boldsymbol{\sigma}=\operatorname{softmax}(\mathbf{z})$. The matrix $\operatorname{diag}(\boldsymbol{\sigma})-\boldsymbol{\sigma}\boldsymbol{\sigma}^T$ is the covariance matrix of a categorical distribution with probabilities $\boldsymbol{\sigma}$, which has maximum eigenvalue at most $1/2$ (achieved when $K=2$ and $\sigma_1=\sigma_2=1/2$). Therefore,
\begin{equation}
L_{\mathrm{sm}}=\frac{1}{2\tau}.
\end{equation}

\textbf{Composition.}
By the chain rule, the overall Lipschitz constant satisfies
\begin{multline}
L_{\mathcal{F}}=L_{\mathrm{sm}}\cdot L_{\mathrm{MLP}}\cdot L_\gamma\cdot L_{\mathrm{norm}}\\
=\frac{1}{2\tau}\left(\prod_{l=1}^{D}\|W_l\|_2\right)\sqrt{1+\frac{\pi^2(4^L-1)}{3}}\;\max_i\frac{1}{\Delta_i}.
\end{multline}
\end{proof}

\begin{corollary}
The material variation between two positions is bounded by
\begin{equation}
\lVert\mathbf{m}(\mathbf{x}_1)-\mathbf{m}(\mathbf{x}_2)\rVert_2\leq M_p\sqrt{K}\cdot L_{\mathcal{F}}\,\lVert\mathbf{x}_1-\mathbf{x}_2\rVert_2,
\end{equation}
where $M_p=\max_k\lVert\mathbf{m}_k\rVert_2$ is the maximum palette entry norm.
\end{corollary}

\begin{proof}
Since $\mathbf{m}(\mathbf{x})=\sum_k w_k(\mathbf{x})\mathbf{m}_k$,
\begin{align}
\lVert\mathbf{m}(\mathbf{x}_1)-\mathbf{m}(\mathbf{x}_2)\rVert_2&=\left\lVert\sum_k\bigl(w_k(\mathbf{x}_1)-w_k(\mathbf{x}_2)\bigr)\mathbf{m}_k\right\rVert_2\nonumber\\
&\leq\lVert\mathbf{w}(\mathbf{x}_1)-\mathbf{w}(\mathbf{x}_2)\rVert_1\cdot M_p.
\end{align}
Converting between norms, $\lVert\cdot\rVert_1\leq\sqrt{K}\,\lVert\cdot\rVert_2$, and applying the Lipschitz bound on $\mathbf{w}$ yields the result.
\end{proof}

This bound has a natural interpretation: nearby Gaussians receive similar palette assignments because the material field is a continuous function of position, and the bound tightens as $\tau$ decreases (sharper assignments). The dominant factor is typically $L_{\mathrm{MLP}}$, which training keeps moderate.

\subsection{SH vs.\ Path-Traced Indirect Illumination}

During inverse rendering (Stage~2), we approximate indirect illumination using the pre-trained SH radiance field rather than differentiating through expensive multi-bounce path tracing. Figure~\ref{fig:sh_pt_compare} compares this SH approximation against full path tracing on the Chair scene. The low FLIP error confirms that the SH cache faithfully reproduces indirect illumination for single-object scenes, validating its use as a lightweight substitute during optimization. Once materials are decomposed, forward relighting with edited materials uses path tracing instead to correctly capture global illumination effects such as inter-object light transport and material-dependent interreflection.

\begin{figure}[t]
  \centering
  \includegraphics[width=\columnwidth]{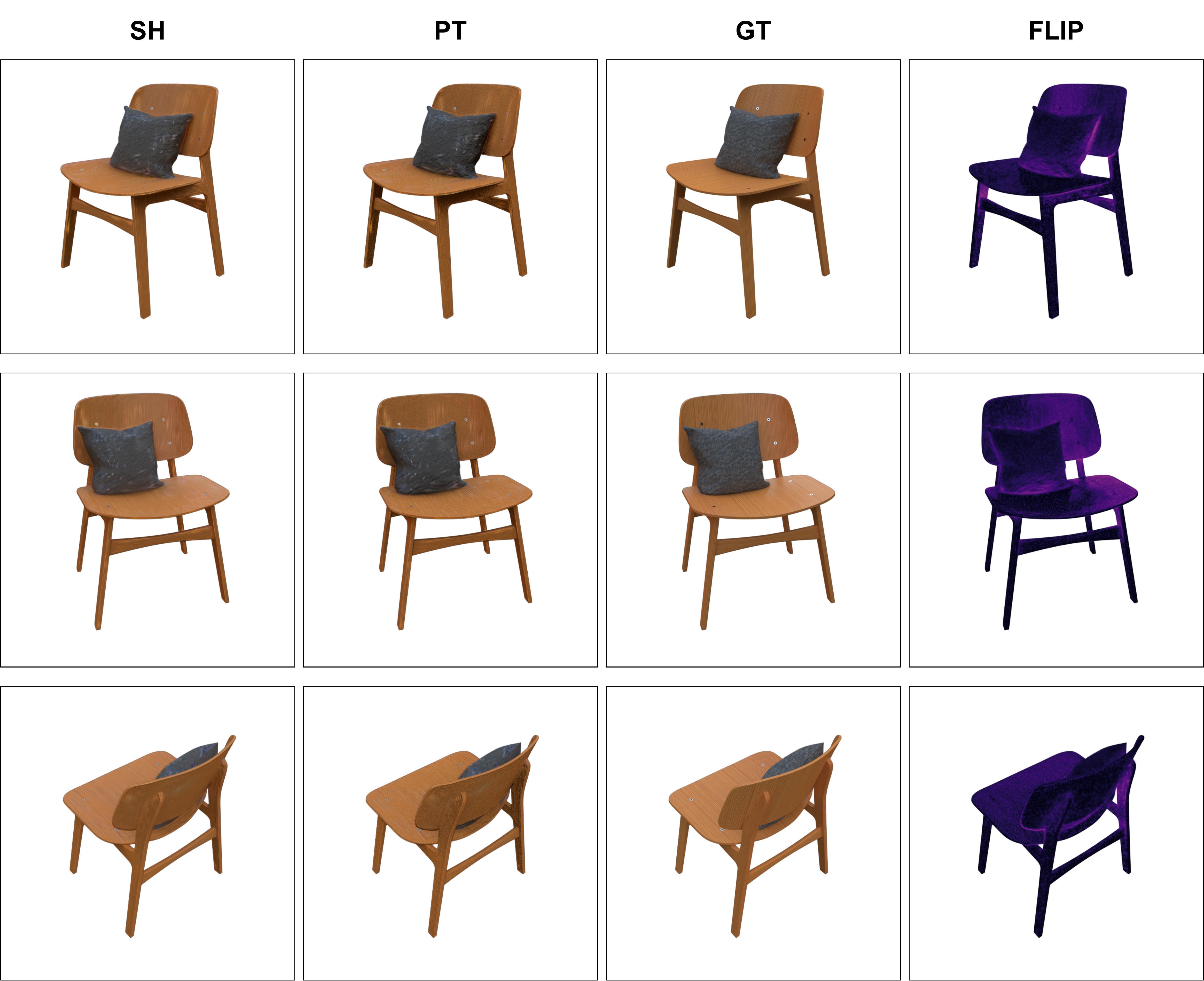}
  \caption{SH-based vs.\ path-traced indirect illumination on the Chair scene, compared via FLIP~\cite{Andersson2020}.}
  \label{fig:sh_pt_compare}
\end{figure}

\subsection{Failure Cases}

\begin{figure}[!t]
  \centering
  \setlength{\abovecaptionskip}{2pt}
  \setlength{\belowcaptionskip}{0pt}
  \includegraphics[width=\columnwidth]{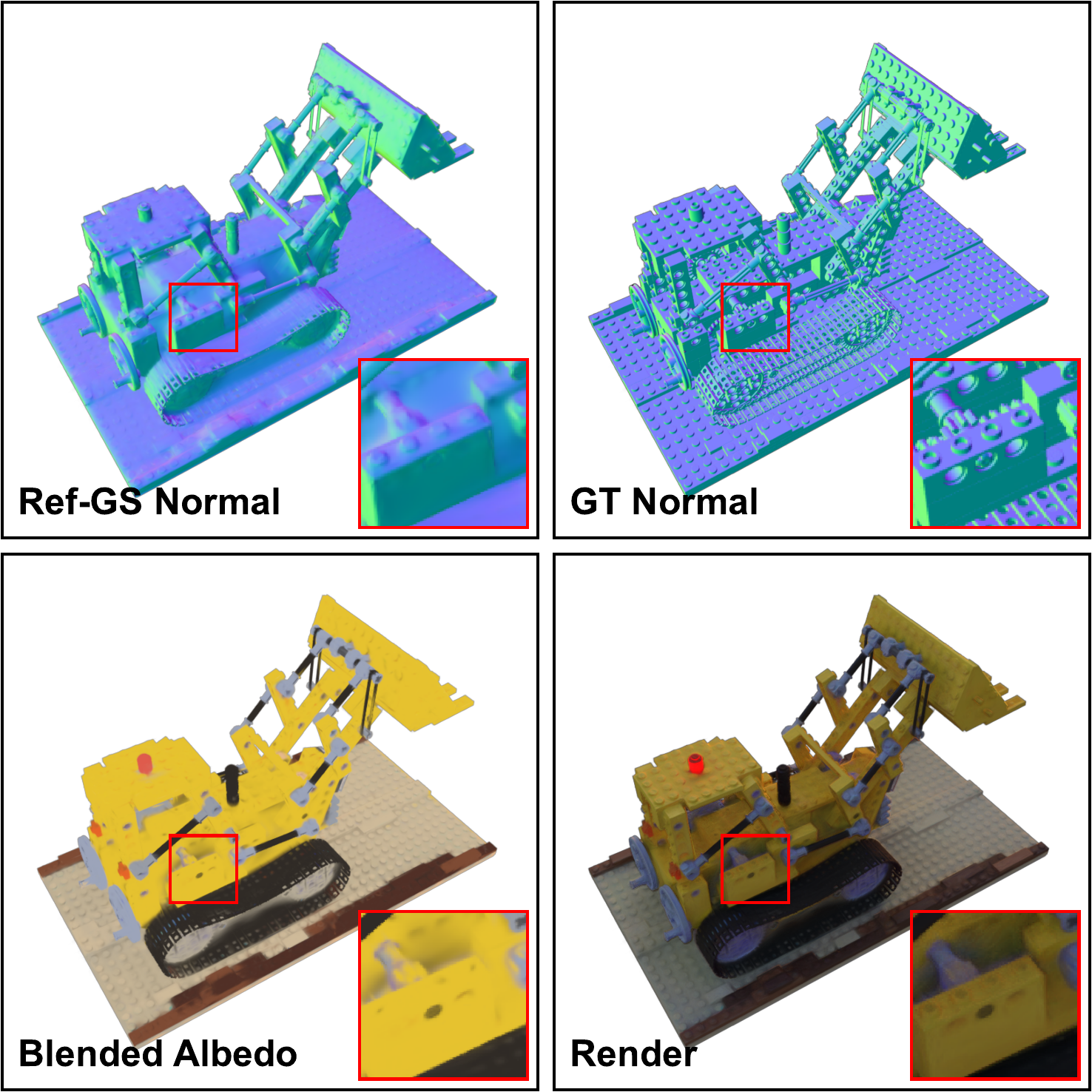}
  \caption{Geometry-induced failure case on the Lego scene. The reconstructed normals from Ref-GS (top left) deviate from the ground truth normals (top right), causing the PBR shading to produce incorrect lighting. The per-material refinement stage compensates by baking the residual shading into the albedo (bottom left), yielding a visually plausible final render (bottom right) at the cost of physical interpretability.}
  \label{fig:failure_geometry}
\end{figure}

\begin{figure}[!t]
  \centering
  \setlength{\abovecaptionskip}{2pt}
  \setlength{\belowcaptionskip}{0pt}
  \includegraphics[width=\columnwidth]{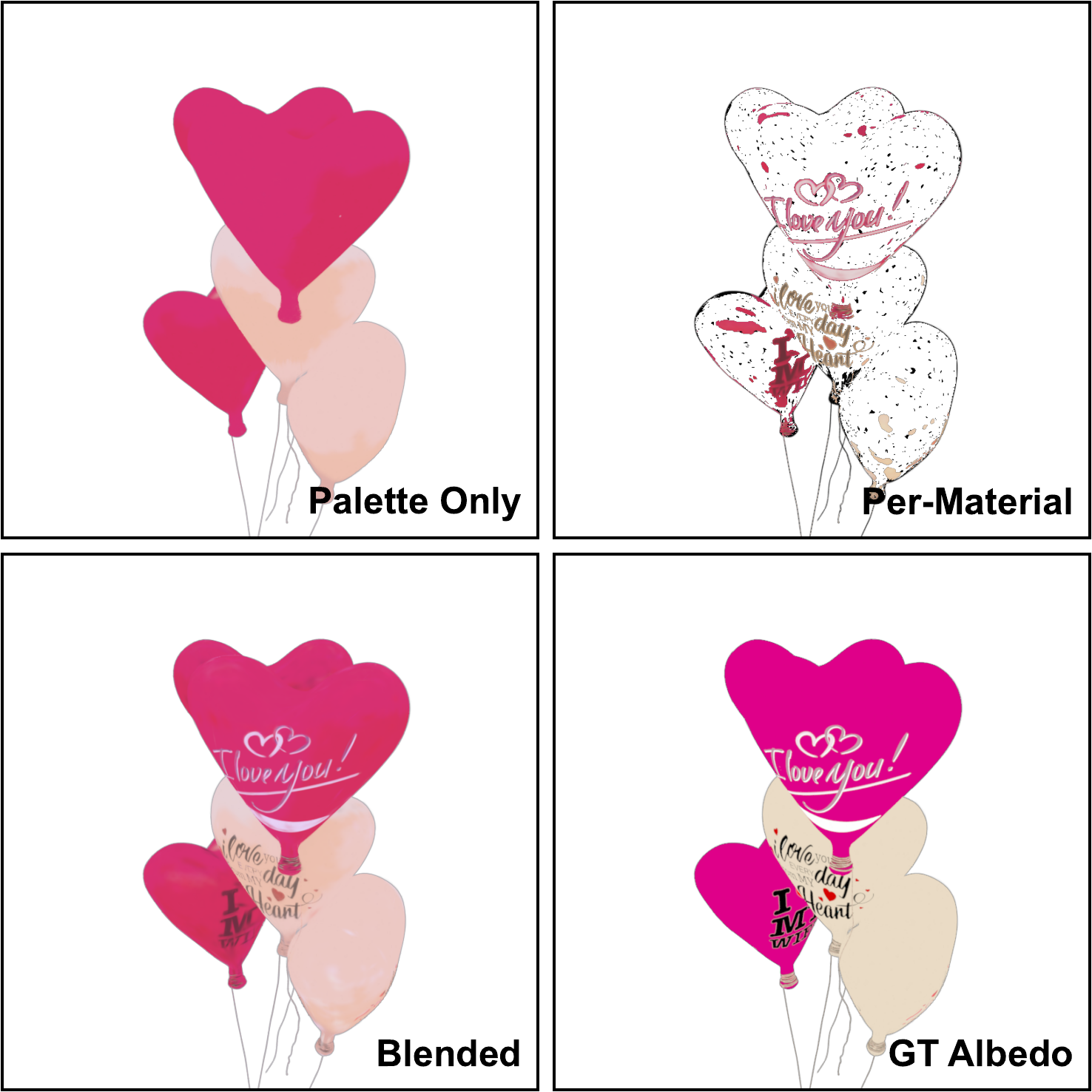}
  \caption{Failure case on the Air Balloons scene ($\rho{=}24\%$). The palette-only albedo (top left) is overly smooth, while the blended albedo (top right) partially compensates via SH mixing. The per-material refined albedo (bottom left) recovers some detail but still deviates from the ground truth albedo (bottom right) in regions with complex spatially-varying textures. This scene is the primary cause of our lower albedo SSIM and LPIPS on the Synthetic4Relight dataset (Table~2), as the fine texture details exceed the representational capacity of a compact material palette.}
  \label{fig:failure_texture}
\end{figure}

\paragraph{Geometry-induced artifacts.}
When the reconstructed Gaussian geometry is inaccurate---misaligned normals, missing surfaces, or incorrect opacity coverage---the physically based shading model receives erroneous inputs. As shown in Figure~\ref{fig:failure_geometry}, the normals recovered by Ref-GS deviate substantially from the ground truth, causing the PBR model to produce incorrect shading. Shadows and ambient occlusion computed from incorrect visibility become baked into the albedo, and the residual weight estimation stage cannot fully separate these geometric artifacts from genuine material variation. In these cases, the per-Gaussian material refinement stage (Stage~4) compensates by re-learning the residual lighting effects into the per-material albedo, yielding a visually plausible final render but at the cost of reduced physical interpretability---the refined materials are no longer purely intrinsic. This behavior is consistent with other Gaussian inverse rendering methods, which similarly absorb unexplained lighting into per-primitive material parameters when the geometry or BRDF model is insufficient.

\paragraph{Complex spatially-varying textures.}
The palette representation assumes that scene materials cluster into a small number of shared prototypes. Surfaces with high-frequency texture detail---wood grain, fabric weave, weathered stone, or printed patterns---exhibit continuous spatial variation that cannot be adequately represented by soft combinations of a few palette entries. In these cases the decomposition either over-segments the scene into many small palette entries (losing the compactness advantage) or produces overly smooth material maps that lose important texture detail. The material residual field mitigates this to some extent, but its capacity is deliberately limited to prevent degenerate solutions.

\subsection{Path Traced Indirect Illumination}

\begin{figure}[t]
  \centering
  \includegraphics[width=\columnwidth]{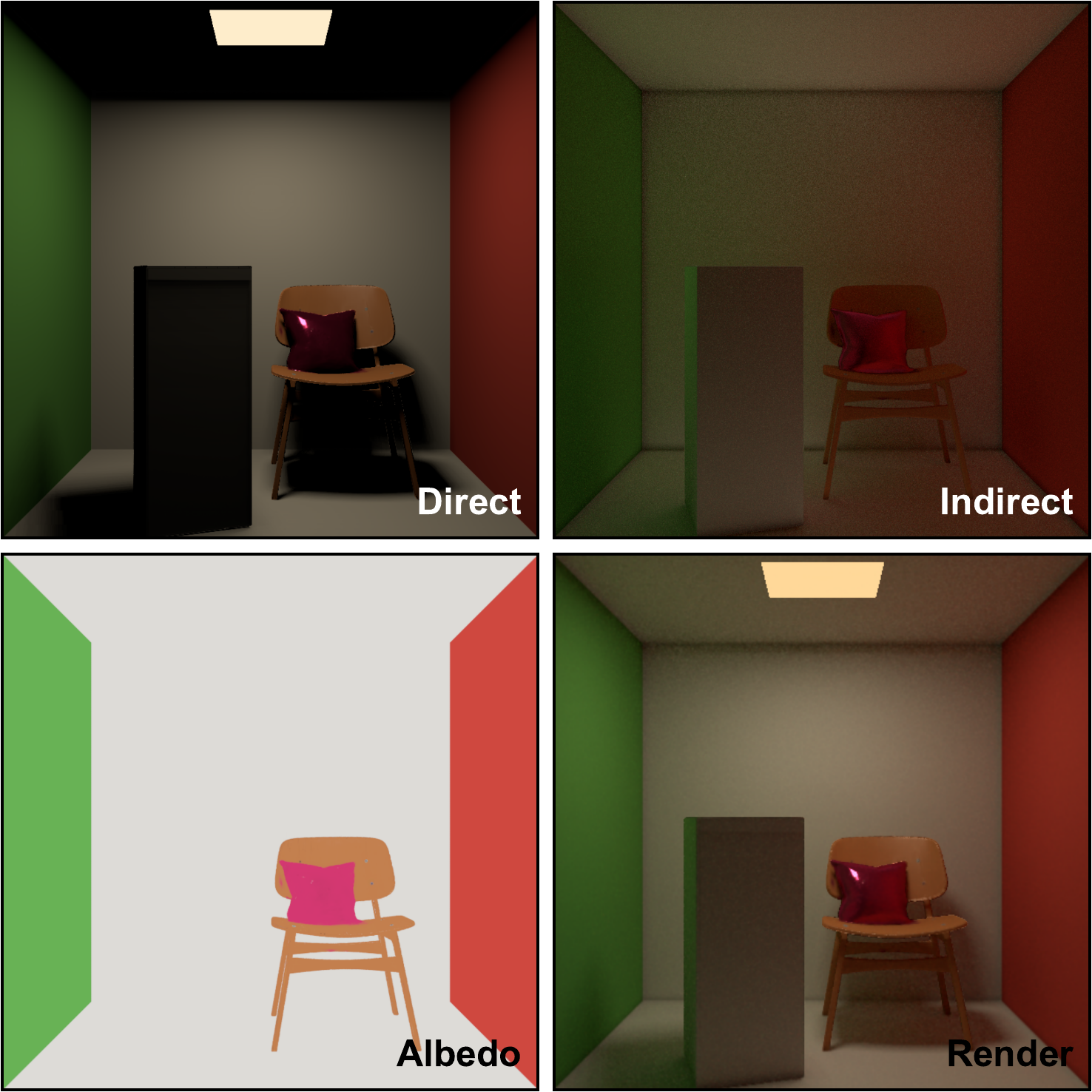}
  \caption{Path traced indirect illumination. A chair with modified palette materials is placed inside a Cornell Box constructed entirely from 2D Gaussians and rendered with our wavefront path tracer. From left to right: direct illumination, indirect illumination (one-bounce global illumination), albedo, and the combined render. Color bleeding from the red and green walls onto the chair surface is correctly captured by the path tracer.}
  \label{fig:cornell_chair}
\end{figure}

During inverse rendering, indirect illumination is approximated by the pre-trained SH radiance field (Section~4.5). Once the materials are decomposed, however, the SH cache no longer reflects the modified scene and must be replaced by proper path tracing for relighting. Figure~\ref{fig:cornell_chair} demonstrates this pipeline: we modify the palette material of a reconstructed chair and place it inside a Cornell Box constructed entirely from 2D Gaussians, then render the scene with our wavefront path tracer. The direct component captures primary illumination from the area light, while the indirect component shows one-bounce global illumination including color bleeding from the red and green walls onto the chair surface. Unlike the SH-based cache, the path tracer correctly accounts for inter-object light transport and material-dependent interreflection, confirming that the palette-based decomposition recovers physically meaningful BRDF parameters rather than merely fitting appearance.

\subsection{Albedo Comparison with Insets}

\begin{figure*}[t]
  \centering
  \includegraphics[width=0.95\textwidth]{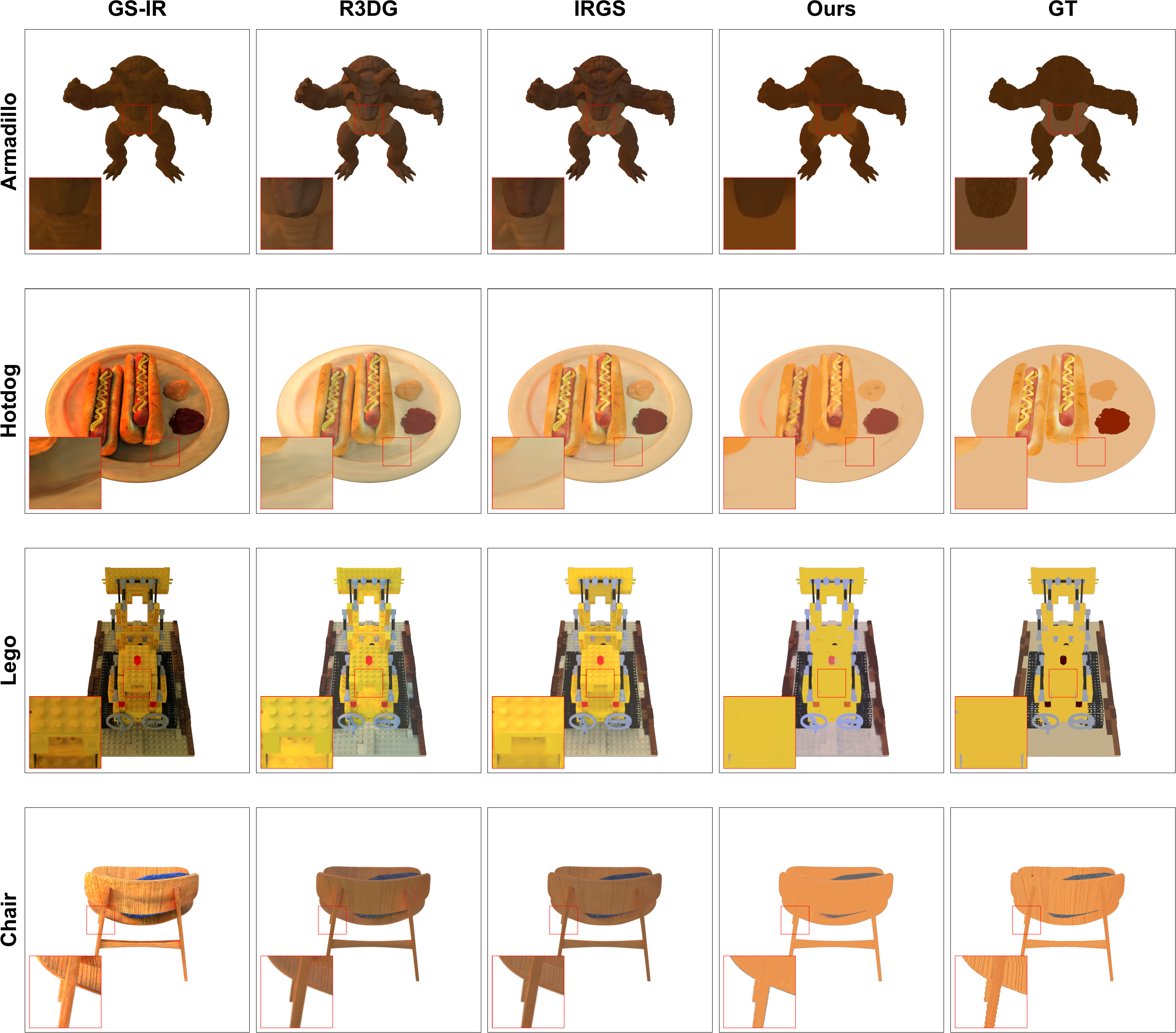}
  \caption{Zoomed-in albedo comparison with insets (see below for discussion).}
  \label{fig:albedo_comparison_inset}
\end{figure*}

Figure~\ref{fig:albedo_comparison_inset} provides a zoomed-in comparison of albedo decomposition across all scenes. The insets highlight a key advantage of the palette-based representation: because all Gaussians assigned to the same palette entry share identical BRDF parameters, the recovered albedo is spatially consistent and free from the high-frequency noise and local artifacts that characterize per-primitive methods such as GS-IR, R3DG, and IRGS. These baselines often bake residual shading, contact shadows, or visibility errors into individual albedo estimates, producing slight color variations across regions that should share the same material. In contrast, our method forces a single albedo prototype to explain all Gaussians of the same material, naturally suppressing such artifacts.

This consistency is essential for palette-level editing: modifying one palette entry produces a uniform appearance change across all associated regions, which is not guaranteed when each Gaussian carries independent material parameters. However, this compactness comes at a cost. As discussed in the failure cases (Figures~\ref{fig:failure_geometry}--\ref{fig:failure_texture}), the palette cannot capture high-frequency texture detail in regions where material appearance varies continuously at a fine scale (e.g., the Chair scene), leading to slightly over-smooth albedo estimates in such areas.